\documentclass[entropy,article,accept,oneauthor,pdftex,12pt,a4paper]{mdpi}

\usepackage{makeidx}
\usepackage{graphicx}
\usepackage{amssymb,amsmath,amsthm}
\usepackage{hyphenat}
\usepackage[sort&compress]{natbib}
\input epsf
\newcommand{\D}{{\mathrm{d}}}
\newtheorem{theorem}{Theorem}

\newtheorem{proposition}{Proposition}
\newtheorem{definition}{Definition}

\lastpage{????}
\doinum{10.3390/e16052408}
\pubvolume{16}
\pubyear{2014}
\history{Received: 9 March 2014; in revised form: 15 April 2014 / Accepted: 24 April 2014 /
\linebreak Published: 29 April 2014 / Postprint with new Appendix 03 June 2014 }

\Title{General $H$-theorem and Entropies that Violate the Second Law}
\Author{Alexander N. Gorban}
\address{Department of Mathematics, University of Leicester, University Road, Leicester, LE1 7RH, UK; E-Mail: ag153@le.ac.uk}
\corres{.}

\abstract{$H$-theorem states that the entropy production is nonnegative and, therefore, the entropy of
a closed system should monotonically change in time. In information processing, the entropy production
is positive for random transformation of signals (the information processing lemma). Originally,
the $H$-theorem and the information processing lemma were proved for the classical
Boltzmann-Gibbs-Shannon entropy and for the correspondent divergence (the relative entropy).
Many new entropies and divergences have been proposed during last decades and for all of them the
$H$-theorem is needed. This note proposes a simple and general criterion to check whether
the $H$-theorem is valid for a convex divergence $H$ and demonstrates that some of the popular divergences obey
no $H$-theorem. We consider systems with $n$ states $A_i$ that obey first order kinetics (master equation).
A convex function $H$ is a Lyapunov function for all master equations with given equilibrium if and only if
its conditional minima properly describe the equilibria of pair transitions $A_i \rightleftharpoons A_j$.
This theorem does not depend on the principle of detailed balance and is valid for general Markov kinetics.
Elementary analysis of pair equilibria demonstrates that the popular Bregman divergences like Euclidean distance
or Itakura-Saito distance in the space of distribution cannot be the universal Lyapunov functions for the
first-order kinetics and can increase in Markov processes. Therefore, they violate the second law and the information
processing lemma. In particular, for these measures of information (divergences) random manipulation with data
may add information to data. The main results are extended to nonlinear generalized mass action law kinetic equations.
In Appendix, a new family of the universal Lyapunov functions for the generalized mass action law kinetics is described.}

\begin{document}

\section{The Problem}

The first non-classical entropy was proposed by R\'enyi in 1960 \cite{Renyi1961}. In the
same paper he discovered the very general class of divergences, the so-called
$f$-divergences (or Csisz\'ar-Morimoto divergences because of the works of Csisz\'ar
\cite{Csiszar1963} and Morimoto \cite{Morimoto1963} published simultaneously in 1963):
\begin{equation}\label{Morimoto}
H_h(p)=H_h(P \| P^*)=\sum_i p^*_i h\left(\frac{p_i}{p_i^*}\right)
\end{equation}
where $P=(p_i)$ is a probability distribution, $P^*$ is an equilibrium distribution,
$h(x)$ is a convex function defined on the open ($x>0$) or closed $x\geq 0$ semi-axis. We
use here the notation $H_h(P \| P^*)$ to stress the dependence of $H_h$ on both $p_i$ and
$p^*_i$.

These divergences have the form of the relative entropy or, in the thermodynamic
terminology, the (negative) free entropy, the Massieu-Planck functions \cite{Callen1985},
or $F/RT$ where $F$ is the free energy. They measure the deviation of the current
distribution $P$ from the equilibrium $P^*$.

After 1961, many new entropies and divergences were invented and applied to real
problems, including Burg entropy \cite{Burg1972}, Cressie-Red family of power divergences
\cite{CR1984}, Tsallis entropy \cite{Tsa1988,Abe}, families of $\alpha$-, $\beta$- and
$\gamma$-divergences \cite{CichockiAmari2010} and many others (see the review papers
\cite{EstMor1995,GorGorJudge2010}). Many of them have the $f$-divergence form, but some
of them do not. For example, the squared Euclidean distance from $P$ to $P^*$ is not, in
general, a $f$-divergence unless all $p^*_i$ are equal (equidistribution). Another
example gives the Itakura-Saito distance:
\begin{equation}\label{Itakura-Saito}
\sum_i \left(\frac{p_i}{p_i^*}-\ln \frac{p_i}{p_i^*}- 1\right)
\end{equation}
The idea of Bregman divergences \cite{Bregman1967} provides a new general source of
divergences different from the $f$-divergences. Any strictly convex function $F$ in an
closed convex set $V$ satisfies the Jensen inequality
\begin{equation}
D_F(p,q)=F(p)-F(q)-(\nabla_q F(q), p-q) >0
\end{equation}
if $p\neq q$, $p,q \in V$. This positive quantity $D_F(p,q)$ is the Bregman divergence
associated with $F$. For example, for a positive quadratic form $F(x)$ the Bregman
distance is just $D_F(p,q)=F(p-q)$. In particular, if $F$ is the squared Euclidean length
of $x$ then $D_F(p,q)$ is the squared Euclidean distance. If $F$ is the Burg entropy,
$F(x)=-\sum_i \ln p_i$, then $D_F(p,q)$ is the Itakura-Saito distance. The Bregman
divergences have many attractive properties. For example, the mean vector minimizes the
expected Bregman divergence from the random vector \cite{Banerjee2005}. The Bregman
divergences are convenient for numerical optimization because generalized Pythagorean
identity \cite{CsizarMatus2012}. Nevertheless, for information processing and for many
physical applications one more property is crucially important.

The divergence between the current distribution and equilibrium should monotonically
decrease in Markov processes. It is the ultimate requirement for use of the divergence in
information processing and in non-equilibrium thermodynamics and kinetics. In physics,
the first result of this type was Boltzmann's $H$-theorem proven for nonlinear kinetic
equation. In information theory, Shannon \cite{Shannon1948} proved this theorem for the
entropy (``the data processing lemma") and Markov chains.

In his well-known paper \cite{Renyi1961}, R\'enyi also proved that $H_h(P \| P^*)$
monotonically decreases in Markov processes (he gave the detailed proof for the classical
relative entropy and then mentioned that for the $f$-divergences it is the same). This
result, elaborated further by Csisz\'ar \cite{Csiszar1963} and Morimoto
\cite{Morimoto1963}, embraces many later particular $H$-theorems for various entropies
including the Tsallis entropy and the R\'enyi entropy (because it can be transformed into
the form (\ref{Morimoto}) by a monotonic function, see for
example~\cite{GorGorJudge2010}). The generalized data processing lemma was proven
\cite{Cohen1993,CohenIwasa1993}: for every two positive probability distributions $P,Q$
the divergence $H_h(P \| Q)$ decreases under action of a stochastic matrix $A=(a_{ij})$
\begin{equation}
H_h(AP \| AQ)\leq \overline{\alpha}(A) H_h(P \| Q)
\end{equation}
where
\begin{equation}
\overline{\alpha}(A)=\frac{1}{2}\max_{i,k} \left\{\sum_j |a_{ij}-a_{kj}| \right\}
\end{equation}
is the ergodicity contraction coefficient, $0 \leq \overline{\alpha}(A) \leq 1$. Here,
neither $Q$ nor $P$ must be the equilibrium distribution: divergence between any two
distributions decreases in Markov processes.

Under some additional conditions, the property to decrease in Markov processes
characterizes the $f$-divergences \cite{ENTR3,Amari2009}. For example, if a divergence
decreases in all Markov processes, does not change under permutation of states and can be
represented as a sum over states (has the trace form), then it is the $f$-divergence
\cite{ENTR3,GorGorJudge2010}.

The dynamics of distributions in the continuous time Markov processes is described by the
master equation. Thus, the $f$-divergences are the Lyapunov functions for the master
equation. The important property of the divergences $H_h(P \| P^*)$ is that they are {\em
universal} Lyapunov functions. That is, they depend on the current distribution $P$ and
on the equilibrium $P^*$ but do not depend on the transition probabilities directly.

For each new divergence we have to analyze its behavior in Markov processes and to prove
or refute the $H$-theorem. For this purpose, we need a simple and general criterion. It
is desirable to avoid any additional requirements like the trace form or symmetry. In
this paper we develop this criterion.

It is obvious that the equilibrium $P^*$ is a global minimum of any universal Lyapunov
function $H(P)$ in the simplex of distributions (see the model equation below). In brief,
the {\em general $H$-theorem} states that a convex function $H(P)$ is a universal
Lyapunov function for the master equation if and only if its conditional minima correctly
describe the partial equilibria for pairs of transitions $A_i \rightleftharpoons A_j$.
These partial equilibria are given by proportions $p_i/p_i^*=p_j/p_j^*$. They should be
solutions to the problem
\begin{equation}\label{ConditMin}
\begin{split}
&H(P)\rightarrow \min \; \mbox{ subject to } p_k \geq 0 \, (k=1,\ldots, n), \\
& \sum_{k=1}^n p_k=1, \;
\mbox{ and given values of } p_l \, (l\neq i,j)
\end{split}
\end{equation}
These solutions are minima of $H(P)$ on segments $p_i+p_j=1-\sum_{l\neq i,j} p_l$,
$p_{i,j}\geq0$. They depend on $n-2$ parameters $p_l\geq 0$ ($l\neq i,j$, $\sum_{l\neq
i,j} p_l<1$).

Using this general $H$-theorem we analyze several Bregman divergences that are not
$f$-divergences and demonstrate that they do not allow the $H$-theorem even for systems
with three states. We present also the generalizations of the main results for
Generalized Mass Action Law (GMAL) kinetics.

\section{Three Forms of Master Equation and the Decomposition Theorem}

We consider continuous time Markov chains with $n$ states $A_1, \ldots , A_n$. The {\em
Kolmogorov equation or master equation} for the probability distribution $P$ with the
coordinates $p_i$ (we can consider $P$ as a vector-column $P=[p_1,\ldots, p_n]^T$) is
\begin{equation}\label{MAsterEq0}
\frac{\D p_i}{\D t}= \sum_{j, \, j\neq i} (q_{ij}p_j-q_{ji}p_i) \;\; (i=1,\ldots, n)
\end{equation}
where $q_{ij}$ ($i,j=1,\ldots, n$, $i\neq j$) are nonnegative. In this notation, $q_{ij}$
is the {\em rate constant} for the transition $A_j \to A_i$. Any set of nonnegative
coefficients $q_{ij}$ ($i\neq j$) corresponds to a master equation. Therefore, the class
of the master equations can be represented as a nonnegative orthant in
$\mathbb{R}^{n(n-1)}$ with coordinates $q_{ij}$ ($i\neq j$). Equations of the same class
describe any first order kinetics in perfect mixtures. The only difference between the
general first order kinetics and master equation for the probability distribution is in
the balance conditions: the sum of probabilities should be 1, whereas the sum of
variables (concentrations) for the general first order kinetics may be any positive
number.

It is useful to mention that the {\em model equation} with equilibrium $P^*$ and
relaxation time $\tau$
\begin{equation}\label{model}
\frac{\D p_i}{\D t}= \frac{1}{\tau}(p_i^*-p_i) \;\; (i=1,\ldots, n)
\end{equation}
is a particular case of master equation for normalized variables $p_i$ ($p_i\geq 0$,
$\sum_i p_i=1$). Indeed, let us take in Equation (\ref{MAsterEq0})
$q_{ij}=\frac{1}{\tau}p_i^*$.

The {\em graph of transitions} for a Markov chain is a directed graph. Its vertices
correspond to the states $A_i$ and the edges correspond to the transitions $A_j \to A_i$
with the positive transition coefficients, $q_{ij}>0$. The digraph of transitions is {\em
strongly connected} if there exists an oriented path from any vertex $A_i$ to every other
vertex $A_j$ ($i\neq j$). The continuous-time Markov chain is {\em ergodic} if there
exists a unique strictly positive equilibrium distribution $P^*$ ($p^*_i>0$, $\sum_i
p^*_i=1$) for master equation (\ref{MAsterEq0}) \cite{MeynNets2007,MeynMarkCh2009}.
Strong connectivity of the graph of transitions is necessary and sufficient for
ergodicity of the corresponding Markov chain.

A digraph is {\em weakly connected} if the underlying undirected graph obtained by
replacing directed edges by undirected ones is connected. The maximal weakly connected
components of a digraph are called connected (or weakly connected) components. The
maximal strongly connected subgraphs are called strong components. The necessary and
sufficient condition for the {\em existence} of a strongly positive equilibrium for
master equation (\ref{MAsterEq0}) is: {\em the weakly connected components of the
transition graph are its strong components}. An equivalent form of this condition is:
{\em if there exists a directed path from $A_i$ to $A_j$, then there exists a directed
path from $A_j$ to $A_i$}. In chemical kinetics this condition is sometimes called the
``weak reversibility'' condition \cite{Feinberg1974,SzederkHang2012}. This implies that
the digraph is the union of disjoint strongly connected digraphs. For each strong
component of the transition digraph the normalized equilibrium is unique and the
equilibrium for the whole graph is a convex combination of positive normalized equilibria
for its strong components. If $m$ is the number of these components then the set of
normalized positive equilibria of master equation ($P^*$: $p^*_i>0$, $\sum_i p^*_i=1$) is
a relative interior of a $m-1$-dimensional polyhedron in the unit simplex $\Delta_n$. The
set of non-normalized positive equilibria ($P^*$: $p^*_i>0$) is a relative interior of a
$m$-dimensional cone in the positive orthant $\mathbb{R}_+^n$.

We reserve notation $\mathbb{R}_+^n$ for the positive orthant and for the nonnegative
orthant we use $\overline{\mathbb{R}_+^n}$ (the closure of $\mathbb{R}_+^n$)

The Markov chain in Equation (\ref{MAsterEq0}) is {\em weakly ergodic} if it allows the
only conservation law: the sum of coordinates, $\sum_i p_i \equiv const$. Such a system
forgets its initial condition: the distance between any two trajectories with the same
value of the conservation law tends to zero when time goes to infinity. Among all
possible norms, the $l_1$ distance ($\|P-Q \|_{l_1}=\sum_i |p_i-q_i|$) plays a special
role: it does not increase in time for any first order kinetic system in master equation
(\ref{MAsterEq0}) and strongly monotonically decreases to zero for normalized probability
distributions ($\sum_i p_i=\sum_i q_i=1$) and weakly ergodic chains. The difference
between weakly ergodic and ergodic systems is in the obligatory existence of a {\em
strictly positive} equilibrium for an ergodic system. A Markov chain is weakly ergodic if
and only if for each two vertices $A_i, \: A_j \: (i \neq j)$ we can find such a vertex
$A_k$ that is reachable by oriented paths both from $A_i$ and from $A_j$. This means that
the following structure exists \cite{GorbanPathSumm2011}:
\begin{equation}\label{elementBridge}
A_i \to \ldots \to A_k \leftarrow \ldots \leftarrow A_j \ .
\end{equation}
One of the paths can be degenerated: it may be $i=k$ or $j=k$.

Now, let us restrict our consideration to the set of the Markov chains with the given
positive equilibrium distribution $P^*$ ($p^*_i>0$). We do not assume that this
distribution is compulsory unique. The transition graph should be the union of disjoint
strongly connected digraphs (in particular, it may be strongly connected). Using the
known positive equilibrium $P^*$ we can rewrite master equation (\ref{MAsterEq0}) in the
following form
\begin{equation}\label{MAsterEq1}
\frac{\D p_i}{\D t}= \sum_{j, \, j\neq i}
q_{ij}p^*_j\left(\frac{p_j}{p_j^*}-\frac{p_i}{p_i^*}\right) \
\end{equation}
where $p_i^*$ and $q_{ij}$ are connected by the {\em balance equation}
\begin{equation}\label{MasterEquilibrium}
\sum_{j, \, j\neq i} q_{ij}p^*_j = \left(\sum_{j, \, j\neq i}
q_{ji}\right)p^*_i \; \mbox{ for all }i=1,\ldots, n
\end{equation}

For the next transformation of master equation we join the mutually reverse transitions
in pairs $A_i \rightleftharpoons A_j$ in pairs (say, $i>j$) and introduce the {\em
stoichiometric vectors} $\gamma^{ji}$ with coordinates:
\begin{equation}\label{gamma}
\gamma^{ji}_k=\left\{\begin{array}{ll}
-1 &\mbox{ if } k=j,\\
1 &\mbox{ if } k=i,\\
0 &\mbox{ otherwise}
\end{array}
\right.
\end{equation}
Let us rewrite the master equation (\ref{MAsterEq0}) in the {\em quasichemical form}:
\begin{equation}\label{QuasiChemKolgen}
\frac{\D P}{\D t}=\sum_{i>j}(w^+_{ij}-w^-_{ij})\gamma^{ji}
\end{equation}
where $w_{ij}^+=q_{ij}{p_j^*}\frac{p_j}{p_j^*}$ is the rate of the transitions $A_j \to
A_i$ and $w_{ij}^-=q_{ji}{p_i^*}\frac{p_i}{p_i^*}$ is the rate of the reverse process
$A_j \leftarrow A_i$ ($i>j$).

The reversible systems with detailed balance form an important class of first order
kinetics. The {\em detailed balance} condition reads \cite{VanKampen}: at equilibrium,
$w^+_{ij}=w^-_{ij}$, {\em i.e.},
\begin{equation}\label{detBal}
q_{ij}p^*_j=q_{ji}p^*_i \,(=w_{ij}^*) \;\; i,j=1,\ldots, n
\end{equation}
Here, $w_{ij}^*$ is the {\em equilibrium flux} from $A_i$ to $A_j$ and back.

For the systems with detailed balance the quasichemical form of the master equation is
especially simple:
\begin{equation}\label{QuasiChemKol}
\frac{\D P}{\D t}=\sum_{i>j}w_{ij}^*\left(\frac{p_j}{p_j^*}-\frac{p_i}{p_i^*}\right) \gamma^{ji}
\end{equation}
It is important that any set of nonnegative equilibrium fluxes $w_{ij}^*$ ($i>j$) defines
by Equation (\ref{QuasiChemKol}) a system with detailed balance with a given positive
equilibrium $P^*$. Therefore, the set of all systems with detailed balance presented by
Equation (\ref{QuasiChemKol}) and a given equilibrium may be represented as a nonnegative
orthant in $\mathbb{R}^{\frac{n(n-1)}{2}}$ with coordinates $w_{ij}^*$ ($i>j$).

The {\em decomposition theorem} \cite{GorbanEq2012arX,GorbanCAMWA2013} states that for
any given positive equilibrium $P^*$ and any positive distribution $P$ the set of
possible values ${\D P}/{\D t}$ for Equations (\ref{QuasiChemKolgen}) under the balance
condition (\ref{MasterEquilibrium}) coincides with the set of possible values ${\D P}/{\D
t}$ for Equations (\ref{QuasiChemKol}) under detailed balance condition (\ref{detBal}).

In other words, for every general system of the form (\ref{QuasiChemKolgen}) with
positive equilibrium $P^*$ and any given non-equilibrium distribution $P$ there exists a
system with detailed balance of the form (\ref{QuasiChemKol}) with the same equilibrium
and the same value of the velocity vector ${\D P}/{\D t}$ at point $P$. Therefore, the
sets of the universal Lyapunov function for the general master equations and for the
master equations with detailed balance coincide.

\section{General $H$-Theorem}

Let $H(P)$ be a convex function on the space of distributions. It is a Lyapunov function
for a master equations with the positive equilibrium $P^*$ if $\D H(P(t))/ \D t \leq 0$
for any positive normalized solution $P(t)$. For a system with detailed balance given by
Equation (\ref{QuasiChemKol})
\begin{equation}\label{entropyprodDelBal}
\frac{\D H(P(t))}{\D t}=-\sum_{i>j}w_{ij}^*\left(\frac{p_j}{p_j^*}-\frac{p_i}{p_i^*}\right)
\left(\frac{\partial H(P)}{\partial p_j}-\frac{\partial H(P)}{\partial p_i}\right)
\end{equation}
The inequality $\D H(P(t))/ \D t \leq 0$ is true for all nonnegative values of $w_{ij}^*$
if and only is it holds for any term in Equation (\ref{entropyprodDelBal}) separately.
That is, for any pair $i,j$ ($i>j$) the convex function $H(P)$ is a Lyapunov function for
the system (\ref{QuasiChemKol}) where one and only one $w_{ij}^*$ is not zero.

A convex function on a straight line is a Lyapunov function for a one-dimensional system
with single equilibrium if and only if the equilibrium is a minimizer of this function.
This elementary fact together with the previous observation gives us the criterion for
universal Lyapunov functions for systems with detailed balance. Let us introduce the {\em
partial equilibria criterion}:
\begin{definition}[Partial equilibria criterion]\label{def:PatEqCrit}A convex function $H(P)$ on the simplex $\Delta_n$ of probability distributions satisfies the partial equilibria criterion with a positive equilibrium $P^*$ if the proportion $p_i/p_i^*=p_j/p_j^*$ give the minimizers in the problem (\ref{ConditMin}).
\end{definition}
\begin{proposition}\label{prop:DBMarkH}A convex function $H(P)$ on the simplex $\Delta_n$ of probability distributions is a Lyapunov function for all master equations with the given equilibrium $P^*$ that obey the principle of detailed balance if and only if it satisfies the partial equilibria criterion with the equilibrium $P^*$.
\end{proposition}
Combination of this Proposition with the decomposition theorem \cite{GorbanEq2012arX}
gives the same criterion for general master equations without hypothesis about detailed
balance
\begin{proposition}\label{prop:CBMarkH}
A convex function $H(P)$ on the simplex $\Delta_n$ of probability distributions is a
Lyapunov function for all master equations with the given equilibrium $P^*$ if and only
if it satisfies the partial equilibria criterion with the equilibrium $P^*$.
\end{proposition}
These two propositions together form the general $H$-theorem.
\begin{theorem}\label{theor:MarkH}The partial equilibria criterion with a positive equilibrium $P^*$ is a necessary condition for a convex function to be the universal Lyapunov function for all master equations with detailed balance and equilibrium $P^*$ and a sufficient condition for this function to be the universal Lyapunov function for all master equations with equilibrium $P^*$.
\end{theorem}
Let us stress that here the partial equilibria criterion provides a necessary condition
for systems with detailed balance (and, therefore, for the general systems without
detailed balance assumption) and a sufficient condition for the general systems (and,
therefore, for the systems with detailed balance too).

\begin{figure}[t]\centering{
\includegraphics[height=0.35\textwidth]{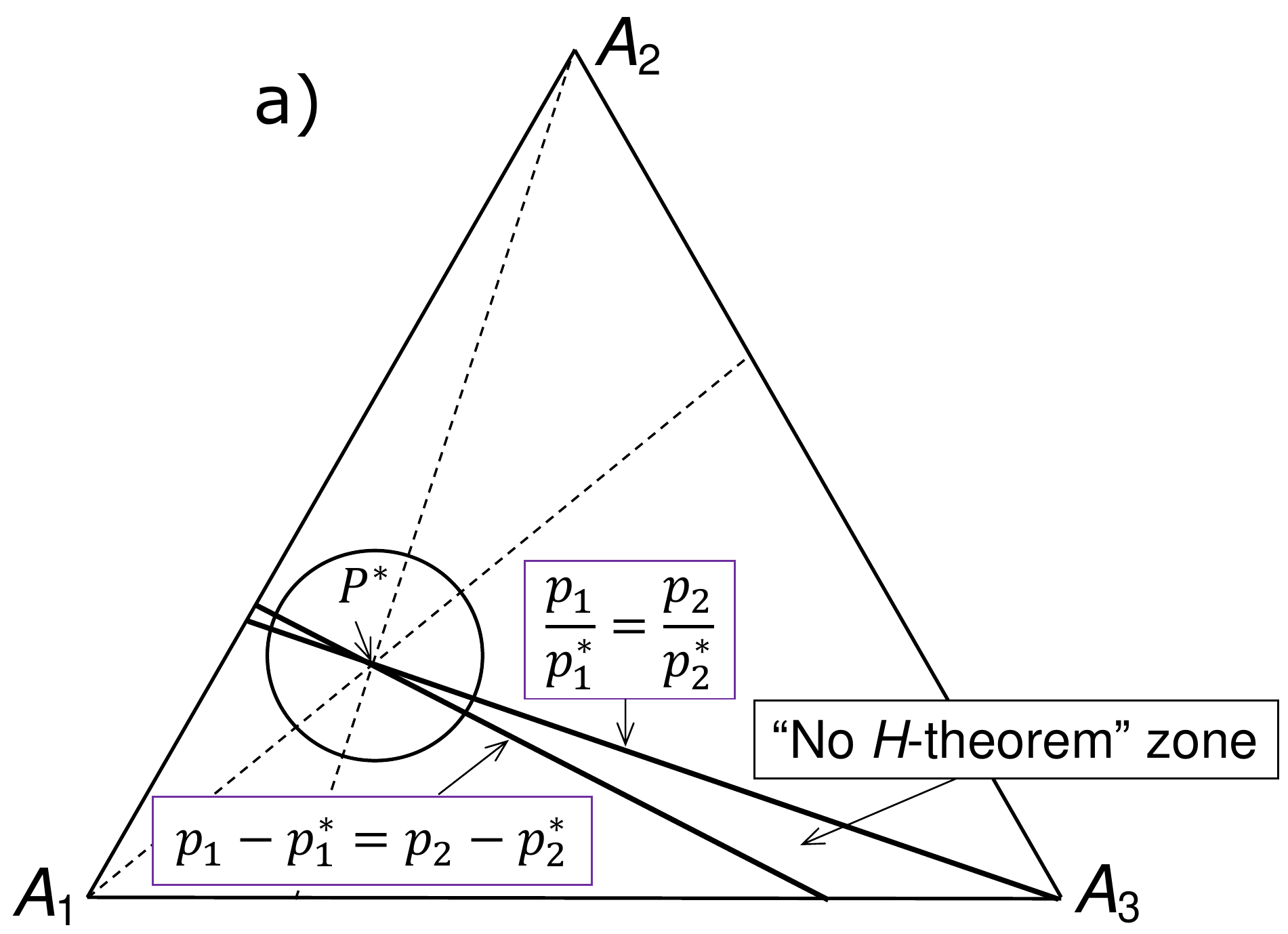} \hspace{5mm}
\includegraphics[height=0.35\textwidth]{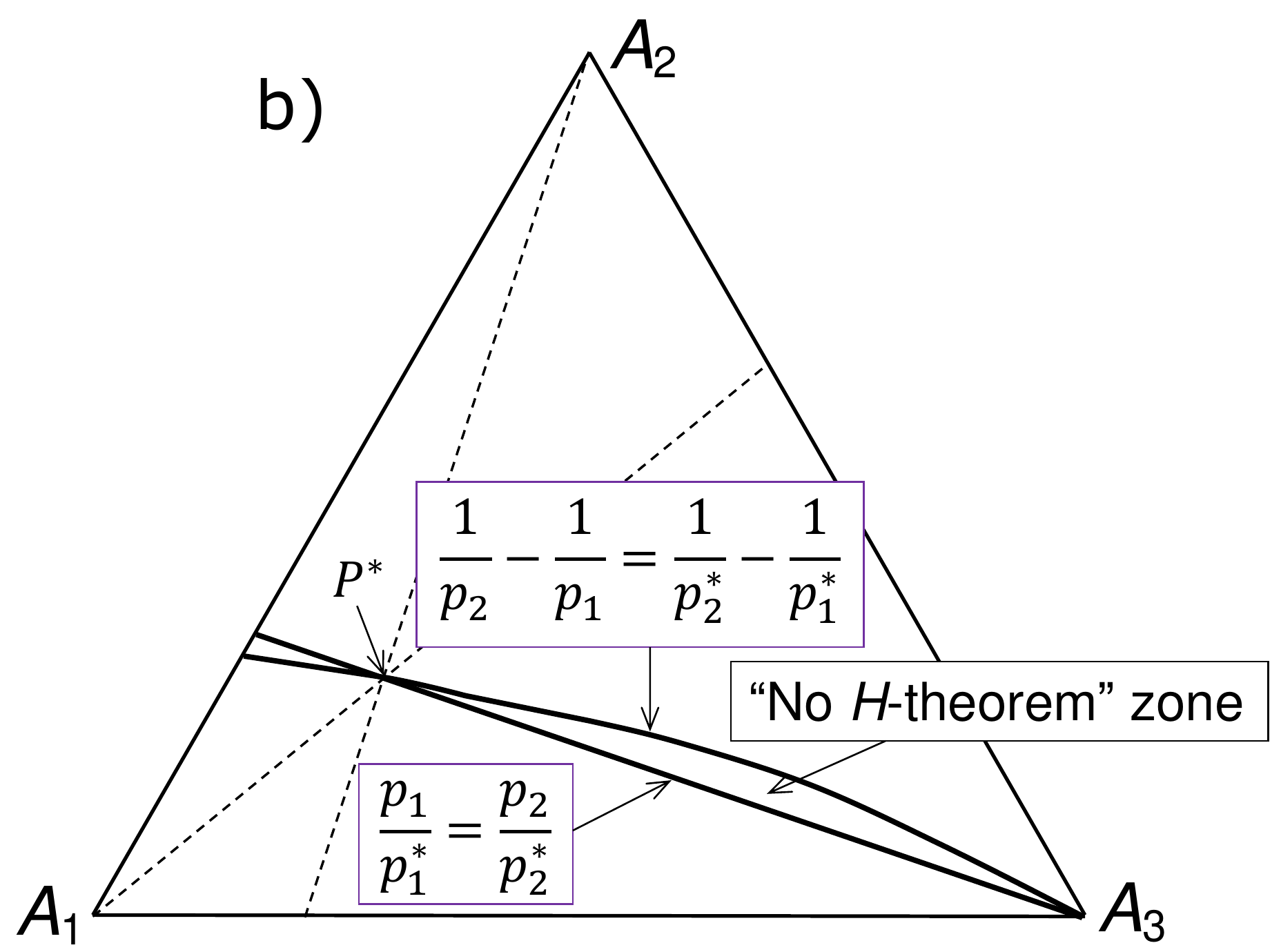}}
\caption{The triangle of distributions for the system with three states $A_1$, $A_2$, $A_3$ and the equilibrium $p_1^*=\frac{4}{7}$, $p_2^*=\frac{2}{7}$, $p_3^*=\frac{1}{7}$. The lines of partial equilibria $A_i \rightleftharpoons A_j$ given by the proportions $p_i/p_i^*=p_j/p_j^*$ are shown, for $A_1 \rightleftharpoons A_2$ by solid straight lines (with one end at the vertex $A_3$), for $A_2 \rightleftharpoons A_3$ and for $A_1 \rightleftharpoons A_3$ by dashed lines. The lines of conditional minima of $H(P)$ in problem (\ref{ConditMin}) are presented for the partial equilibrium $A_1 \rightleftharpoons A_2$ (a) for the squared Euclidean distance (a circle here is an example of the $H(P)$ level set) and (b) for the Itakura-Saito distance. Between these lines and the line of partial equilibria the ``no $H$-theorem zone" is situated. In this zone, $H(P)$ increases in time for some master equations with equilibrium $P^*$. Similar zones (not shown) exist near other partial equilibrium lines too. Outside these zones, $H(P)$ monotonically decreases in time for any master equation with equilibrium $P^*$. \label{fig:examples}}
\end{figure}

\section{Examples}

The simplest Bregman divergence is the squared Euclidean distance between $P$ and $P^*$,
$\sum_i (p_i-p_i^*)^2$. The solution to the problem (\ref{ConditMin}) is:
$p_i-p_i^*=p_j-p_j^*$. Obviously, it differs from the proportion required by the partial
equilibria criterion $\frac{p_i}{p_j}=\frac{p_i^*}{p_j^*}$ (Figure~\ref{fig:examples}a).

For the Itakura-Saito distance (\ref{Itakura-Saito}) the solution to the problem
(\ref{ConditMin}) is: $\frac{1}{p_i}-\frac{1}{p_i^*}=\frac{1}{p_j}-\frac{1}{p_j^*}$. It
also differs from the proportion required (Figure~\ref{fig:examples}b).

If the single equilibrium in 1D system is not a minimizer of a convex function $H$ then
$\D H/\D t >0$ on the interval between the equilibrium and minimizer of $H$ (or
minimizers if it is not unique). Therefore, if $H(P)$ does not satisfy the partial
equilibria criterion then in the simplex of distributions there exists an area bordered
by the partial equilibria surface for $A_i \rightleftharpoons A_j$ and by the minimizers
for the problem (\ref{ConditMin}), where for some master equations $\D H/\D t >0$
(Figure~\ref{fig:examples}). In particular, in such an area $\D H/\D t >0$ for the simple
system with two mutually reverse transitions, $A_i \rightleftharpoons A_j$, and the same
equilibrium.

If $H$ satisfies the partial equilibria criterion, then the minimizers for the problem
(\ref{ConditMin}) coincide with the partial equilibria surface for $A_i
\rightleftharpoons A_j$, and the ``no $H$-theorem zone" vanishes.

The partial equilibria criterion allows a simple geometric interpretation. Let us
consider a sublevel set of $H(P)$ in the simplex $\Delta_n$: $U_h=\{P\in \Delta_n \ | \
H(P) \leq h \}$. Let the level set be $L_h=\{P\in \Delta_n \ | \ H(P) = h \}$. For the
partial equilibrium $A_i \rightleftharpoons A_j$ we use the notation $E_{ij}$. It is
given by the equation $p_i/p_i^*=p_j/p_j^*$. The geometric equivalent of the partial
equilibrium condition is: for all $i,j$ ($i\neq j$) and every $P\in L_h \cap E_{ij}$ the
straight line $P+\lambda \gamma_{ij}$ ($\lambda \in \mathbb{R}$) is a supporting line of
$U_h$. This means that this line does not intersect the interior of $U_h$.

We illustrate this condition on the plane for three states in
Figure~\ref{fig:CorrectLevelSet}. The level set of $H$ is represented by the dot-dash
line. It intersects the lines of partial equilibria (dashed lines) at points $B_{1,2,3}$
and $C_{1,2,3}$. For each point $P$ from these six intersections ($P=B_i$ and $P=C_i$)
the line $P+\lambda \gamma_{jk}$ ($\lambda \in \mathbb{R}$) should be a supporting line
of the sublevel set (the region bounded by the dot-dash line). Here, $i,j,k$ should all
be different numbers. Segments of these lines form a hexagon circumscribed around the
level set (\mbox{Figure~\ref{fig:CorrectLevelSet}b}).

The points of intersection $B_{1,2,3}$ and $C_{1,2,3}$ cannot be selected arbitrarily on
the lines of partial equilibria. First of all, they should be the vertices of a convex
hexagon with the equilibrium $P^*$ inside. Secondly, due to the partial equilibria
criterion, the intersections of the straight line $P+\lambda \gamma_{ij}$ with the
partial equilibria $E_{ij}$ are the conditional minimizers of $H$ on this line, and
therefore should belong to the sublevel set $U_{H(P)}$. If we apply this statement to
$P=B_i$ and $P=C_i$, then we will get two projections of this point onto partial
equilibria $E_{ij}$ parallel to $\gamma_{ij}$ (Figure~\ref{fig:CorrectLevelSet}a). These
projections should belong to the hexagon with the vertices $B_{1,2,3}$ and $C_{1,2,3}$.
They produce a six-ray star that should be inscribed into the level set.

In Figure~\ref{fig:CorrectLevelSet} we present the following characterization of the
level set of a Lyapunov function for the Markov chains with three states. This convex set
should be circumscribed around the six-ray star (Figure~\ref{fig:CorrectLevelSet}a) and
inscribed in the hexagon of the supporting lines (Figure~\ref{fig:CorrectLevelSet}b).

\begin{figure}[t]\centering{
\includegraphics[height=0.35\textwidth]{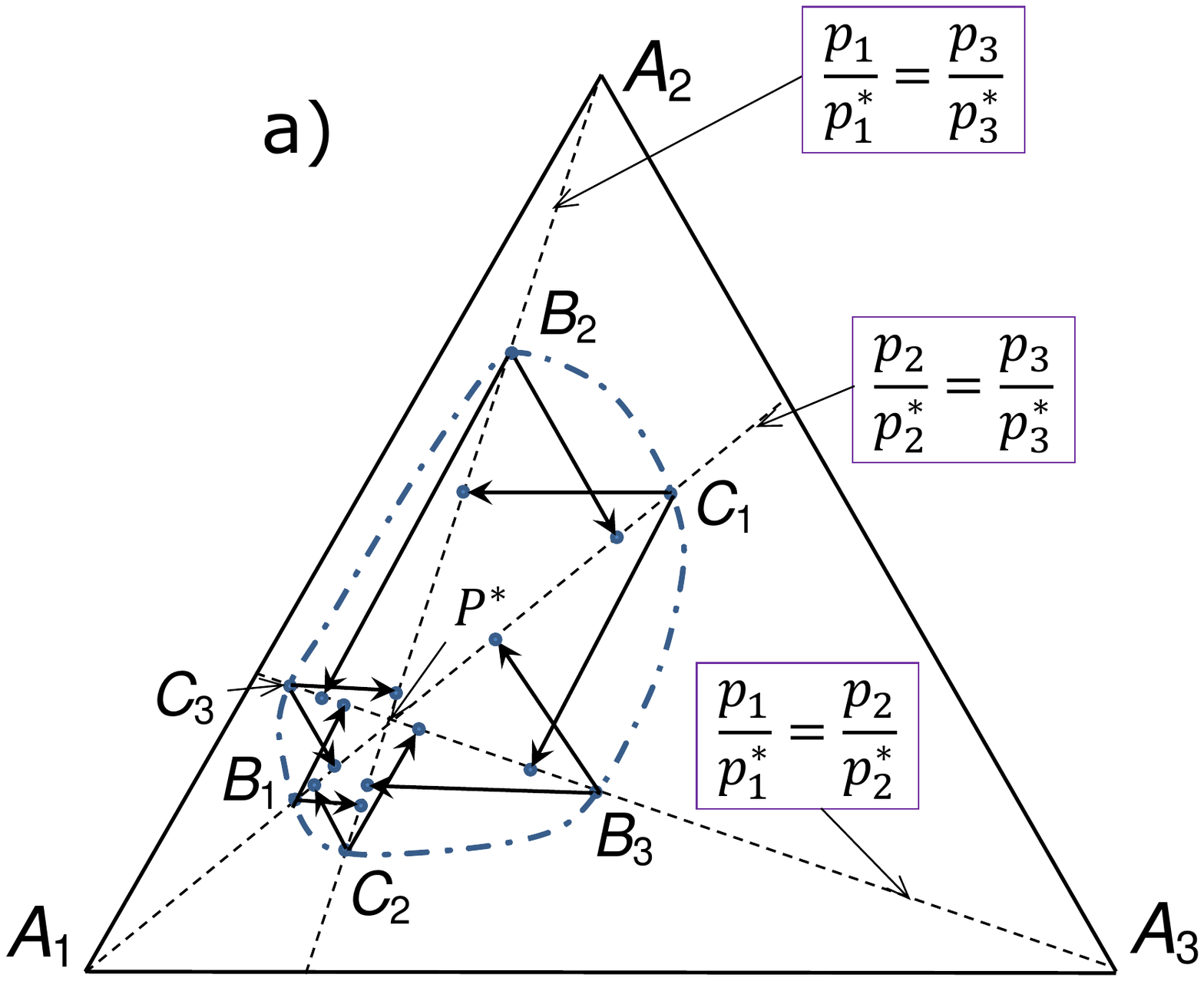} \hspace{5mm}
\includegraphics[height=0.35\textwidth]{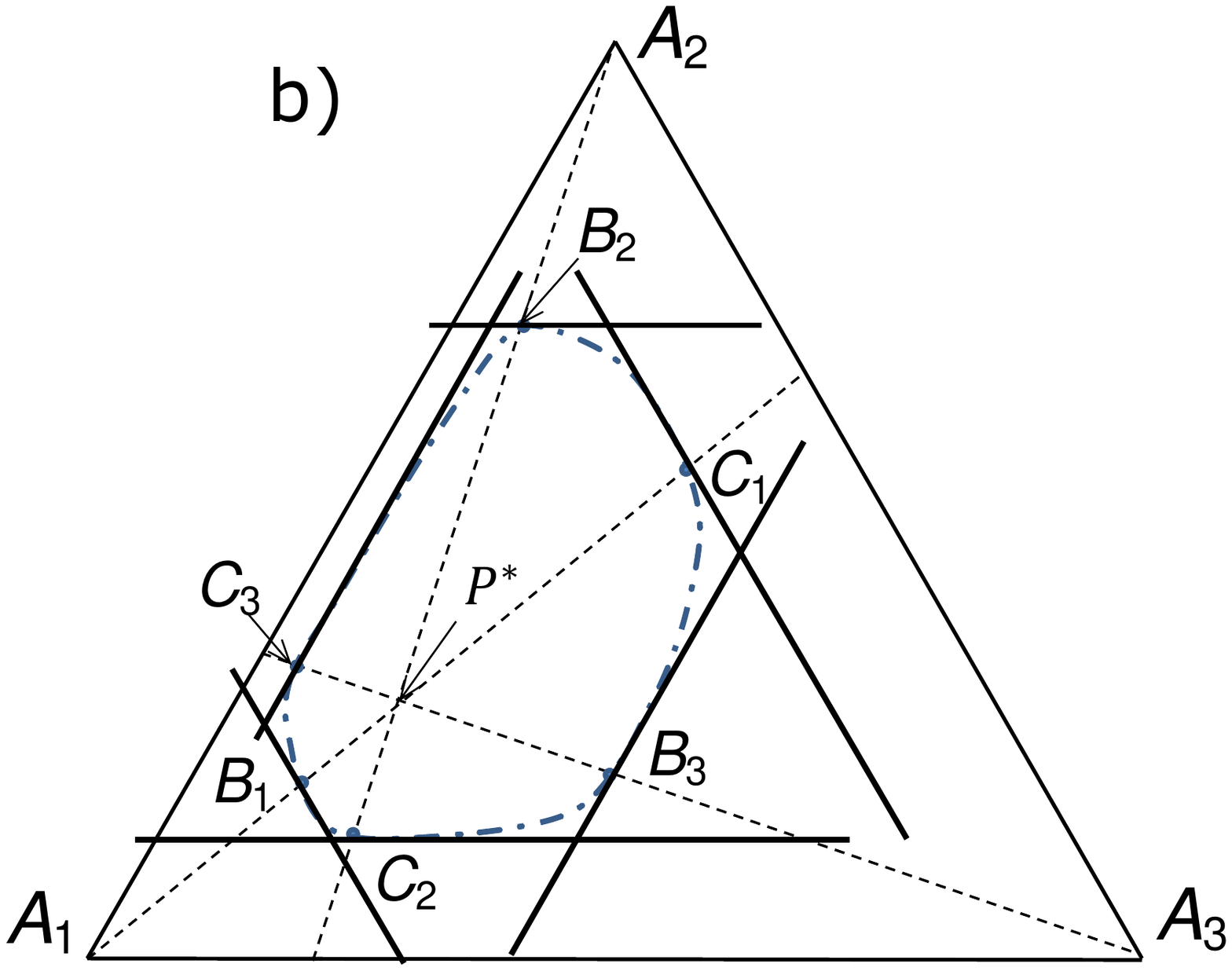}}
\caption{{\it Geometry of the Lyapunov function level set.} The triangle of distributions for the system
with three states $A_1$, $A_2$, $A_3$ and the equilibrium $p_1^*=\frac{4}{7}$, $p_2^*=\frac{2}{7}$, $p_3^*=\frac{1}{7}$.
The lines of partial equilibria $A_i \rightleftharpoons A_j$ given by the proportions $p_i/p_i^*=p_j/p_j^*$ are shown by dashed lines.
The dash-dot line is the level set of a Lyapunov function $H$. It intersects the lines of partial equilibria at points $B_{1,2,3}$ and
$C_{1,2,3}$. (The points $B_i$ are close to the vertices $A_i$, the points $C_i$ belong to the same partial equilibrium but
on another side of the equilibrium $P^*$.) For each point $B_i$, $C_i$ the corresponding partial equilibria of two transitions
$A_i \rightleftharpoons A_j$ ($j\neq i$) are presented (a). These partial equilibria should belong to the sublevel set of $H$.
They are the projections of $B_i$, $C_i$ onto the lines of partial equilibria $A_i \rightleftharpoons A_j$ ($j\neq i$)
with projecting rays parallel to the sides $[A_i,A_j]$ of the triangle ({\em i.e.}, to the
stoichiometric vectors $\gamma^{ji}$ (\ref{gamma})). The six-ray star with vertices $B_i$, $C_i$ should be inside the dash-dot contour (a).
Therefore, the projection of $B_i$ onto the partial equilibrium $A_i \rightleftharpoons A_j$ should belong to the segment $[C_k,P^*]$
and the projection of $C_i$ onto the partial equilibrium $A_i \rightleftharpoons A_j$ should belong to the segment $[B_k,P^*]$ (a).
The lines parallel to the sides $A_j,A_k$ of the triangle should be supporting lines of the level set of $H$ at points $B_i$, $C_i$ ($i,j,k$
are different numbers) (b). Segments of these lines form a circumscribed hexagon around the level set (b).\label{fig:CorrectLevelSet}}
\end{figure}

All the $f$-divergences given by Equation (\ref{Morimoto}) satisfy the partial equilibria
criterion and are the universal Lyapunov functions but the reverse is not true: the class
of universal Lyapunov functions is much wider than the set of the $f$-divergences. Let us
consider the set ``$PEC$'' of convex functions $H(P\|P^*)$, which satisfy the partial
equilibria criterion. It is closed with respect to the following operations
\begin{itemize}
\item{Conic combination: if $H_j(P\|P^*)\in PEC$ then $\sum_j \alpha_j H_j(P\|P^*)\in
    PEC$ for nonnegative coefficients $\alpha_j \geq0$.}
\item{Convex monotonic transformation of scale: if $H(P\|P^*)\in PEC$ then
    $F(H(P\|P^*))\in PEC$ for any convex monotonically increasing function of one
    variable $F$.}
\end{itemize}
Using these operations we can construct new universal Lyapunov functions from a given
set. \linebreak For example, $$\frac{1}{2}\sum_i\frac{(p_i-p_i^*)^2}{p_i^*}+ \prod_j
\exp\frac{(p_j-p_j^*)^2}{2p_j^*}$$ is a universal Lyapunov function that does not have
the $f$-divergence form because the first sum is an $f$-divergence given by Equation
(\ref{Morimoto}) with $h(x)=\frac{1}{2}(x-1)^2$ and the product is the exponent of this
$f$-divergence (exp is convex and monotonically increasing function).

The following function satisfies the partial equilibria criterion for every
$\varepsilon>0$.
\begin{equation}\label{NonclassicExampleEnt}
\frac{1}{2}\sum_i\frac{(p_i-p_i^*)^2}{p_i^*}+ \frac{\varepsilon}{4n^2} \prod_{i,j, i\neq j} {({p_i}{p_j^*}-{p_j}{p_i^*})^2}
\end{equation}
It is convex for $0<\varepsilon<1$. (Just apply the Gershgorin theorem \cite{Golub} to
the Hessian and use that all $p_i,p_i^*\leq 1$.) Therefore, it is a universal Lyapunov
function for master equation in $\Delta_n$ if \mbox{$0<\varepsilon<1$}. The partial
equilibria criterion together with the convexity condition allows us to construct many
such~examples.

\section{General $H$-Theorem for Nonlinear Kinetics}
\vspace{-12pt}

\subsection{Generalized Mass Action Law}

Several formalisms are developed in chemical kinetics and non-equilibrium thermodynamics
for the construction of general kinetic equations with a given ``thermodynamic Lyapunov
functional''. The motivation of this approach ``from thermodynamics to kinetics'' is
simple \cite{G1,YBGE}: (i) the thermodynamic data are usually more reliable than data
about kinetics and we know the thermodynamic functions better than the details of kinetic
equations, and (ii) positivity of entropy production is a fundamental law and we prefer
to respect it ``from scratch'', by the structure of kinetic equations.

GMAL is a method for the construction of dissipative kinetic equations for a given
thermodynamic potential $H$. Other general thermodynamic approaches
\cite{GENERIC,Grmela2012,GiovangigliMatus2012} give similar results for a given
stoichiometric algebra. Below we introduce GMAL following
\cite{G1,YBGE,GorbanShahzad2011}.

The {\em list of components} is a finite set of symbols $A_1, \ldots, A_n$.

A {\em reaction mechanism} is a finite set of the {\em stoichiometric equations} of
elementary reactions:
\begin{equation}\label{stoichiometricequation}
\sum_i\alpha_{\rho i}A_i \to \sum_i \beta_{\rho i} A_i \,
\end{equation}
where $\rho =1, \ldots, m$ is the reaction number and the {\em stoichiometric
coefficients} $\alpha_{\rho i}$, $\beta_{\rho i}$ are nonnegative numbers. Usually, these
numbers are assumed to be integer but in some applications the construction should be
more flexible and admit real nonnegative values. Let $\alpha_{\rho }$, $\beta_{\rho}$ be
the vectors with coordinates $\alpha_{\rho i}$, $\beta_{\rho i}$ correspondingly.

A {\em stoichiometric vector} $\gamma_{\rho}$ of the reaction in Equation
(\ref{stoichiometricequation}) is a $n$-dimensional vector
$\gamma_{\rho}=\beta_{\rho}-\alpha_{\rho}$ with coordinates
\begin{equation}
\gamma_{\rho i}=\beta_{\rho i}-\alpha_{\rho i} \,
\end{equation}
that is, ``gain minus loss'' in the $\rho$th elementary reaction. We assume $\alpha_{\rho
}\neq \beta_{\rho}$ to avoid trivial reactions with zero $\gamma_{\rho}$.

One of the standard assumptions is existence of a strictly positive stoichiometric
conservation law, a vector $b=(b_i)$, $b_i>0$ such that $\sum_i b_i \gamma_{\rho i}=0$
for all $\rho$. This may be the conservation of mass, of the total probability, or of the
total number of atoms, for example.

A nonnegative extensive variable $N_i$, the {\em amount} of $A_i$, corresponds to each
component. We call the vector $N$ with coordinates $N_i$ ``the composition vector''. The
concentration of $A_i$ is an intensive variable $c_i=N_i/V$, where $V>0$ is the volume.
The vector $c=N/V$ with coordinates $c_i$ is the vector of concentrations.

Let us consider a domain $U$ in $n$-dimensional real vector space $\mathbb{R}^n$ with
coordinates $N_1,\ldots, N_n\geq 0$ ($U\subset \overline{\mathbb{R}_+^n}$). For each
$N_i$, a dimensionless entropy (or {\em free entropy}, for example, Massieu, Planck, or
Massieu-Planck potential that corresponds to the selected conditions \cite{Callen1985})
$S(N)$ is defined in $U$. ``Dimensionless'' means that we use $S/R$ instead of physical
$S$. This choice of units corresponds to the informational entropy ($p\log p$ instead of
$k_{\rm B} p\ln p$).

The dual variables, potentials, are defined as the partial derivatives of $H=-S$:
\begin{equation}\label{DualVariables}
\check{\mu}_i=\frac{\partial H}{\partial N_i}, \;\; \check{\mu}=\nabla_N H
\end{equation}
This definition differs from the chemical potentials \cite{Callen1985} by the factor
${1}/{RT}$. We keep the same sign as for the chemical potentials, and this differs from
the standard Legendre transform for $S$. (It is the Legendre transform for the function
$H=-S$.) The standard condition for the reversibility of the Legendre transform is strong
positive definiteness of the Hessian of $H$.

For each reaction, a {\em nonnegative} quantity, reaction rate $r_{\rho}$ is defined. We
assume that this quantity has the following structure (compare with Equations (4), (7),
and (14) in \cite{Grmela2012} and Equation (4.10) in \cite{GiovangigliMatus2012}):
\begin{equation}\label{GeneralReactionRate}
r_{\rho}=\varphi_{\rho}\exp(\alpha_{\rho}, \check{\mu} )
\end{equation}
where $(\alpha_{\rho}, \check{\mu} )=\sum_i \alpha_{\rho i} \check{\mu}_i$ is the
standard inner product. Here and below, $\exp(\; \ , \;)$ is the exponent of the standard
inner product. The {\em kinetic factor} $\varphi_{\rho}\geq $ is an intensive quantity
and the expression $\exp(\alpha_{\rho}, \check{\mu} )$ is the {\em Boltzmann factor} of
the $\rho$th elementary reaction.

In the standard formalism of chemical kinetics the reaction rates are intensive variables
and in kinetic equations for $N$ an additional factor---the volume---appears. For
heterogeneous systems, there may be several ``volumes'' (including interphase surfaces).

A nonnegative extensive variable $N_i$, the amount of $A_i$, corresponds to each
component. We call the vector $N$ with coordinates $N_i$ ``the composition vector''.
$N\in \overline{\mathbb{R}_+^n}$. The concentration of $A_i$ is an intensive variable
$c_i=N_i/V$, where $V>0$ is the volume. If the system is heterogeneous then there are
several ``volumes'' (volumes, surfaces, {\em etc.}), and in each volume there are the
composition vector and the vector of concentrations \cite{YBGE,GorbanShahzad2011}. Here
we will consider homogeneous systems.

The kinetic equations for a homogeneous system in the absence of external fluxes are
\begin{equation}\label{KinUrChem}
\frac{\D N}{\D t}=V \sum_{\rho}r_{\rho} \gamma_{\rho}= V \sum_{\rho} \gamma_{\rho} \varphi_{\rho}\exp(\alpha_{\rho}, \check{\mu} )
\end{equation}
If the volume is not constant then the equations for concentrations include $\dot{V}$ and
have different form (this is typical for combustion reactions, for example).

The classical Mass Action Law gives us an important particular case of GMAL given by
Equation (\ref{GeneralReactionRate}). Let us take the perfect free entropy
\begin{equation}\label{PerfectFreeEntropy}
S=-\sum_iN_i\left(\ln\left(\frac{c_i}{c_i^*}\right)-1\right)
\end{equation}
where $c_i=N_i/V\geq 0$ are concentrations and $c_i^*>0$ are the standard equilibrium
concentrations.

For the perfect entropy function presented in Equation (\ref{PerfectFreeEntropy})
\begin{equation}\label{MALmu}
\check{\mu}_i=\ln\left(\frac{c_i}{c_i^*}\right) \, , \;
\exp(\alpha_{\rho}, \check{\mu} )=\prod_i
\left(\frac{c_i}{c_i^*}\right)^{\alpha_{\rho i}}
\end{equation}
and for the GMAL reaction rate function given by  (\ref{GeneralReactionRate}) we get
\begin{equation}\label{MALreaction rate}
r_{\rho}=\varphi_{\rho}\prod_i
\left(\frac{c_i}{c_i^*}\right)^{\alpha_{\rho i}}
\end{equation}
The standard assumption for the Mass Action Law in physics and chemistry is that
$\varphi$ and $c^*$ are functions of temperature: $\varphi_{\rho}=\varphi_{\rho}(T)$ and
$c^*_i=c^*_i(T)$. To return to the kinetic constants notation and in particular to first
order kinetics in the quasichemical form presented in Equation (\ref{QuasiChemKolgen}),
we should write:
\begin{equation}
\frac{\varphi_{\rho}}{\prod_i {c_i^*}^{\alpha_{\rho i}}}=k_{\rho}
\end{equation}

\subsection{General Entropy Production Formula}

Thus, the following entities are given: the set of components $A_i$ ($i=1,\ldots,n$), the
set of $m$ elementary reactions presented by stoichiometric equations
(\ref{stoichiometricequation}), the thermodynamic Lyapunov function $H(N,V,\ldots)$
\cite{Callen1985,YBGE,Hangos2010}, where dots (marks of omission) stand for the
quantities that do not change in time under given conditions, for example, temperature
for isothermal processes or energy for isolated systems. The GMAL presents the reaction
rate $r_{\rho}$ in Equation (\ref{GeneralReactionRate}) as a product of two factors: the
Boltzmann factor and the kinetic factor. Simple algebra gives for the time derivative of
$H$:
\begin{equation}\label{entropyproductionGNEKIN}
\begin{split}
\frac{\D H}{\D t}&=\sum_i \frac{\partial H}{\partial N_i} \frac{\D N_i}{\D t} \\
&= \sum_i \check{\mu}_i V \sum_{\rho}\gamma_{\rho i} \varphi_{\rho}\exp(\alpha_{\rho},
\check{\mu} ) \\
&=V \sum_{\rho }(\gamma_{\rho }, \check{\mu})\varphi_{\rho}\exp(
\alpha_{\rho }, \check{\mu})
\end{split}
\end{equation}

An auxiliary function $\theta(\lambda) $ of one variable $\lambda\in [0,1]$ is convenient
for analysis of $\D S/ \D t$ (see \mbox{\cite{G1,GorbanShahzad2011,OrlovRozonoer1984}}):
\begin{equation}\label{auxtheta}
\theta(\lambda)=\sum_{\rho}\varphi_{\rho}\exp[(\check{\mu},(\lambda
\alpha_{\rho}+(1-\lambda)\beta_{\rho}))]
\end{equation}
With this function, $\dot{H}$ defined by Equation (\ref{entropyproductionGNEKIN}) has a
very simple form:
\begin{equation}\label{EntropProdtheta}
\frac{\D H}{\D t}=-V\left.\frac{\D \theta(\lambda)}{\D
\lambda}\right|_{\lambda=1}
\end{equation}

The auxiliary function $\theta(\lambda) $ allows the following interpretation. Let us
introduce the deformed stoichiometric mechanism with the stoichiometric vectors,
\begin{equation}
\alpha_{\rho}(\lambda)=\lambda
\alpha_{\rho}+(1-\lambda)\beta_{\rho}\, , \; \beta_{\rho}(\lambda)=\lambda
\beta_{\rho}+(1-\lambda)\alpha_{\rho}
\end{equation}
which is the initial mechanism when $\lambda=1$, the reverted mechanism with interchange
of $\alpha$ and $\beta$ when $\lambda=0$, and the trivial mechanism (the left and right
hand sides of the stoichiometric equations coincide) when $\lambda=1/2$. Let the deformed
reaction rate be $r_{\rho}(\lambda)=\varphi_{\rho}\exp(\alpha_{\rho}(\lambda),
\check{\mu} )$ (the genuine kinetic factor is combined with the deformed Boltzmann
factor). Then $\theta(\lambda)=\sum_{\rho}r_{\rho}(\lambda)$.

It is easy to check that $\theta''(\lambda)\geq 0 $ and, therefore, $\theta(\lambda) $ is
a convex function.

The inequality
\begin{equation}\label{AccordanceIneq}
\theta'(1)\geq 0
\end{equation}
is {\em necessary and sufficient} for accordance between kinetics and thermodynamics
(decrease of free energy or positivity of entropy production). This inequality is a
condition on the kinetic factors. Together with the positivity condition $\varphi_{\rho}
\geq 0$, it defines a convex cone in the space of vectors of kinetic factors
$\varphi_{\rho}$ ($\rho=1,\ldots , m$). There exist two less general and more restrictive
{\em sufficient} conditions: detailed balance and complex balance (known also as
semidetailed or cyclic balance).

\subsection{Detailed Balance}

The detailed balance condition consists of two assumptions: (i) for each elementary
reaction $\sum_i\alpha_{\rho i}A_i \to \sum_i \beta_{\rho i} A_i$ in the mechanism
(\ref{stoichiometricequation}) there exists a reverse reaction $\sum_i\alpha_{\rho i}A_i
\leftarrow \sum_i \beta_{\rho i} A_i$. Let us join these reactions in pairs
\begin{equation}\label{reversibleMechanism}
\sum_i\alpha_{\rho i}A_i \rightleftharpoons \sum_i \beta_{\rho i} A_i
\end{equation}
After this joining, the total number of stoichiometric equations decreases. We
distinguish the reaction rates and kinetic factors for direct and inverse reactions by
the upper plus or minus:
\begin{equation}
r_{\rho}^+=\varphi_{\rho}^+\exp(\alpha_{\rho},\check{\mu})\, , \; r_{\rho}^-=\varphi_{\rho}^-\exp(\beta_{\rho},\check{\mu})
\end{equation}

The kinetic equations take the form
\begin{equation}\label{KinUrChemRev}
\frac{\D N}{\D t}=V \sum_{\rho}(r^+_{\rho}-r^-_{\rho}) \gamma_{\rho}
\end{equation}
The condition of detailed balance in GMAL is simple and elegant:
\begin{equation}\label{detailed balance}
\varphi_{\rho}^+=\varphi_{\rho}^-
\end{equation}
For the systems with detailed balance we can take
$\varphi_{\rho}=\varphi_{\rho}^+=\varphi_{\rho}^-$ and write for the reaction rate:
\begin{equation}\label{DBreactionrate}
r_{\rho}= r^+_{\rho}-r^-_{\rho}=\varphi_{\rho}
(\exp(\alpha_{\rho},\check{\mu})-\exp(\beta_{\rho},\check{\mu}))
\end{equation}
M. Feinberg called this kinetic law the ``Marselin-De Donder'' kinetics
\cite{Feinberg1972_a}.

Under the detailed balance conditions, the auxiliary function $\theta(\lambda)$ is
symmetric with respect to change $\lambda \mapsto (1-\lambda)$. Therefore,
$\theta(1)=\theta(0)$ and, because of convexity of $\theta(\lambda)$, the inequality
holds: $\theta'(1)\geq 0$. Therefore, $\dot{H}\leq 0$ and kinetic equations obey the
second law of thermodynamics.

The explicit formula for $\dot{H}\leq 0$ has the well known form since Boltzmann proved
his $H$-theorem in 1872:
\begin{equation}
\frac{\D H}{\D t}=-V \sum_{\rho} (\ln r_{\rho}^+-\ln r_{\rho}^-) (r_{\rho}^+-r_{\rho}^-)\leq 0
\end{equation}

A convenient equivalent form of $\dot{H}\leq 0$ is proposed in \cite{GorMirYab2013}:
\begin{equation}\label{DissTanh}
\frac{\D H}{\D t}=-V \sum_{\rho} (r_{\rho}^+ + r_{\rho}^-)\mathbb{A}_{\rho} \tanh \frac{\mathbb{A}_{\rho}}{2} \leq 0
\end{equation}
where $$\mathbb{A}_{\rho}=-(\gamma_{\rho},\check{\mu})\;(=-{(\gamma_{\rho},\mu)}/{RT},
\mbox{ where } \mu \mbox{ is the chemical potential})$$ is a normalized {\em affinity}.
In this formula, the kinetic information is collected in the nonnegative factors, the
sums of reaction rates $(r_{\rho}^+ + r_{\rho}^-)$. The purely thermodynamic multiplier
$\mathbb{A}\tanh ({\mathbb{A}}/{2})\geq 0$ is positive for non-zero $\mathbb{A}$. For
small $|\mathbb{A}|$, the expression $\mathbb{A} \tanh ({\mathbb{A}}/{2})$ behaves like
$\mathbb{A}^2/2$ and for large $|\mathbb{A}|$ it behaves like the absolute value,
$|\mathbb{A}|$.

The detailed balance condition reflects ``microreversibility'', that is,
time-reversibility of the dynamic microscopic description and was first introduced by
Boltzmann in 1872 as a consequence of the reversibility of collisions in Newtonian
mechanics.

\subsection{Complex Balance}

The complex balance condition was invented by Boltzmann in 1887 for the Boltzmann
equation \cite{Boltzmann1887} as an answer to the Lorentz objections \cite{Lorentz1887}
against Boltzmann's proof of the $H$-theorem. Stueckelberg demonstrated in 1952 that this
condition follows from the Markovian microkinetics of fast intermediates if their
concentrations are small \cite{Stueckelberg1952}. Under this asymptotic assumption this
condition is just the probability balance condition for the underlying Markov process.
(Stueckelberg considered this property as a consequence of ``unitarity'' in the
$S$-matrix terminology.) It was known as the semidetailed or cyclic balance condition.
This condition was rediscovered in the framework of chemical kinetics by Horn and Jackson
in 1972 \cite{HornJackson1972} and called the complex balance condition. Now it is used
for chemical reaction networks in chemical engineering \cite{SzederkHangos2011}. Detailed
analysis of the backgrounds of the complex balance condition is given in
\cite{GorbanShahzad2011}.

Formally, the complex balance condition means that $\theta(1)\equiv \theta(0)$ for all
values of $\check{\mu}$. We start from the initial stoichiometric equations
(\ref{stoichiometricequation}) without joining the direct and reverse reactions. The
equality $\theta(1)\equiv \theta(0)$ reads
\begin{equation}\label{0=1}
\sum_{\rho}\varphi_{\rho}\exp(\check{\mu},\alpha_{\rho})=
\sum_{\rho}\varphi_{\rho}\exp(\check{\mu},\beta_{\rho})
\end{equation}

Let us consider the family of vectors $\{\alpha_{\rho},\beta_{\rho}\}$ ($\rho=1, \ldots
,m$). Usually, some of these vectors coincide. Assume that there are $q$ different
vectors among them. Let $y_1, \ldots, y_q$ be these vectors. For each $j=1, \ldots, q$ we
take
\begin{equation}
R_j^+=\{\rho\, | \,\alpha_{\rho}=y_j\}\, , \; R_j^-=\{\rho\, |
\,\beta_{\rho}=y_j\}
\end{equation}

We can rewrite Equation (\ref{0=1}) in the form
\begin{equation}\label{0=1incomplex}
\sum_{j=1}^q \exp(\check{\mu},y_j)\left[\sum_{\rho\in R_j^+}\varphi_{\rho}- \sum_{\rho\in R_j^-}\varphi_{\rho}\right]=0
\end{equation}
The Boltzmann factors $\exp(\check{\mu},y_j)$ are linearly independent functions.
Therefore, the natural way to meet these conditions is: for any $j=1, \ldots, q$
\begin{equation}\label{complexbalanceGENKIN}
\sum_{\rho\in R_j^+}\varphi_{\rho}- \sum_{\rho\in R_j^-}\varphi_{\rho}=0
\end{equation}
This is the general {\em complex balance condition}. This condition is sufficient for the
inequality $\dot{H}=\theta'(1) \leq 0$, because it provides the equality
$\theta(1)=\theta(0)$ and $\theta(\lambda)$ is a convex function.

It is easy to check that for the first order kinetics given by Equation (\ref{MAsterEq1})
(or Equation (\ref{QuasiChemKolgen})) with positive equilibrium, the complex balance
condition is just the balance equation (\ref{MasterEquilibrium}) and always holds.

\subsection{Cyclic Decomposition of the Systems with Complex Balance}

The complex balance conditions defined by Equation (\ref{complexbalanceGENKIN}) allow a
simple geometric interpretation. Let us introduce the digraph of transformation of
complexes. The vertices of this digraph correspond to the formal sums $(y,A)$
(``complexes''), where $A$ is the vector of components, and $y\in \{y_1, \ldots, y_q\}$
are vectors $\alpha_{\rho}$ or $\beta_{\rho}$ from the stoichiometric equations of the
elementary reactions (\ref{stoichiometricequation}). The edges of the digraph correspond
to the elementary reactions with non-zero kinetic factor.

Let us assign to each edge $(\alpha_{\rho},A)\to (\beta_{\rho},A)$ the auxiliary {\em
current}---the kinetic factor $\varphi_{\rho}$. For these currents, the complex balance
condition presented by Equation (\ref{complexbalanceGENKIN}) is just Kirchhoff's first
rule: the sum of the input currents is equal to the sum of the output currents for each
vertex. (We have to stress that these auxiliary currents are not the actual rates of
transformations.)

Let us use for the vertices the notation $\Theta_j$: $\Theta_j=(y_j,A)$, ($j=1,\ldots
,q$) and denote $\varphi_{lj}$ the fluxes for the edge $\Theta_j \to \Theta_l$.

The {\em simple cycle} is the digraph $\Theta_{i_1}\to \Theta_{i_2} \to \ldots \to
\Theta_{i_k} \to \Theta_{i_1}$, where all the complexes $\Theta_{i_l}$ ($l=1, \ldots, k$)
are different. We say that the simple cycle is normalized if all the corresponding
auxiliary fluxes are unit: $\varphi_{i_{j+1} \, i_j } = \varphi_{i_{1} \, i_k }=1$.

The graph of the transformation of complexes cannot be arbitrary if the system satisfies
the complex balance condition \cite{Feinberg1974}.

\begin{proposition}\label{prop:Cyclic}If the system satisfies the complex balance condition  (i.e. Equation (\ref{complexbalanceGENKIN}) holds) then every edge of the digraph of transformation of complexes is included into a simple cycle.
\end{proposition}
\begin{proof}
First of all, let us formulate Kirchhoff's first rule (\ref{complexbalanceGENKIN}) for
subsets: if the digraph of transformation of complexes satisfies Equation
(\ref{complexbalanceGENKIN}), then for any set of complexes $\Omega$
\begin{equation}\label{KirghoffSets}
\sum_{\Theta \in \Omega, \Phi \notin \Omega} \varphi_{\Theta \Phi} =\sum_{\Theta \in \Omega, \Phi \notin \Omega} \varphi_{\Phi \Theta}
\end{equation}
where $\varphi_{\Phi \Theta}$ is the positive kinetic factor for the reaction $\Theta\to
\Phi$ if it belongs to the reaction mechanism ({\em i.e.}, the edge $\Theta\to \Phi$
belongs to the digraph of transformations) and $\varphi_{\Phi \Theta}=0$ if it does not.
Equation~(\ref{KirghoffSets}) is just the result of summation of Equations
(\ref{complexbalanceGENKIN}) for all $(y_j,A)=\Theta \in \Omega$.

We say that a state $\Theta_j$ is reachable from a state $\Theta_k$ if $k=i$ or there
exists a non-empty chain of transitions with non-zero coefficients that starts at
$\Theta_k$ and ends at $\Theta_j$: $\Theta_k \to \ldots \to \Theta_j$. Let $\Theta_{i
\downarrow}$ be the set of states reachable from $\Theta_i$. The set $\Theta_{i
\downarrow}$ has no output edges.

Assume that the edge $\Theta_j \to \Theta_i$ is not included in a simple cycle, which
means $\Theta_j \notin \Theta_{i \downarrow}$. Therefore, the set $\Omega=\Theta_{i
\downarrow}$ has the input edge ($\Theta_j \to \Theta_i$) but no output edges and cannot
satisfy Equation (\ref{KirghoffSets}). This contradiction proves the
proposition.\end{proof}

This property (every edge is included in a simple cycle) is equivalent to the so-called
``weak reversibility'' or to the property that every weakly connected component of the
digraph is its strong component.

For every graph with the system of fluxes, which obey Kirchhoff's first rule, the cycle
decomposition theorem holds. It can be extracted from many books and papers
\cite{MeynNets2007,GorbanEq2012arX,Kalpazidou2006}. Let us recall the notion of extreme
ray. A ray with direction vector $x\neq 0$ is a set $\{\lambda x\}$ ($\lambda \geq 0$). A
ray $l$ is an extreme ray of a cone $\mathbf{Q}$ if for any $u \in l$ and any $x,y \in
\mathbf{Q}$, whenever $u = (x + y)/2$, we must have $x,y\in l$. If a closed convex cone
does not include a whole straight line then it is the convex hull of its extreme rays
\cite{Rockafellar1970}.

Let us consider a digraph $Q$ with vertices $\Theta_i$, the set of edges $E$ and the
system of auxiliary fluxes along the edges $\varphi_{ij}\geq 0$ ($(j,i)\in E$). The set
of all nonnegative functions on $E$, $\varphi: (j,i) \mapsto \varphi_{ij}$, is a
nonnegative orthant $\mathbb{R}^{|E|}_+$. Kirchhoff's first rule (Equation
(\ref{complexbalanceGENKIN})) together with nonnegativity of the kinetic factors define a
cone of the systems with complex balance $\mathcal{Q} \subset \mathbb{R}^{|E|}_+$.

\begin{proposition}[Cycle decomposition of systems with complex balance]
\label{prop:Cycle decomposition}Every extreme ray of $\mathcal{Q}$ has a direction vector
that corresponds to a simple normalized cycle $\Theta_{i_1}\to \Theta_{i_2} \to \ldots
\to \Theta_{i_k} \to \Theta_{i_1}$, where all the complexes $\Theta_{i_l}$ ($l=1, \ldots,
k$) are different, all the corresponding fluxes are unit, $\varphi_{i_{j+1} \, i_j } =
\varphi_{i_{1} \, i_k }=1$, and other fluxes are zeros.
\end{proposition}
\begin{proof}
Let a function $\phi:E\to \mathbb{R}_+$ be an extreme ray of $\mathcal{Q}$ and
$\mbox{supp}\phi=\{(j,i)\in E \, | \,\phi_{ij}>0\}$. Due to Proposition~\ref{prop:Cyclic}
each edge from $\mbox{supp}\phi$ is included in a simple cycle formed by edges from
$\mbox{supp}\phi$. Let us take one this cycle $\Theta_{i_1}\to \Theta_{i_2} \to \ldots
\to \Theta_{i_k} \to \Theta_{i_1}$. Denote the fluxes of the corresponding simple
normalized cycle by $\psi$. It is a function on $E$: $\psi_{i_{j+1} \, i_j } =
\psi_{i_{1} \, i_k }=1$ and $\psi_{ij}=0$ if $(i,j)\in E$ but $(i,j)\neq (i_{j+1}, i_j)$
and $(i,j)\neq (i_{1},i_k)$ ($i,j=1,\ldots, k$, $i\neq j$).

Assume that $\mbox{supp}\phi$ includes at least one edge that does not belong to the
cycle $\Theta_{i_1}\to \Theta_{i_2} \to \ldots \to \Theta_{i_k} \to \Theta_{i_1}$. Then,
for sufficiently small $\kappa>0$, $\phi \pm \kappa \psi \in \mathcal{Q}$ and the vector
$\phi \pm \kappa \psi$ is not proportional to $\phi$. This contradiction proves the
proposition.\end{proof}

This decomposition theorem explains why the complex balance condition was often called
the ``cyclic balance condition''.

\subsection{Local Equivalence of Systems with Detailed and Complex Balance}

The class of systems with detailed balance is the proper subset of the class of systems
with complex balance. A simple (irreversible) cycle of the length $k>2$ gives a simplest
and famous example of the complex balance system without detailed balance condition.

For Markov chains, the complex balance systems are all the systems that have a positive
equilibrium distribution presented by Equation (\ref{MasterEquilibrium}), whereas the
systems with detailed balance form the proper subclass of the Markov chains, the
so-called reversible chains.

In nonlinear kinetics, the systems with complex balance provide the natural
generalization of the Markov processes. They deserve the term ``nonlinear Markov
processes'', though it is occupied by a much wider notion \cite{Kolokoltzov}. The systems
with detailed balance form the proper subset of this class.

Nevertheless, in some special sense the classes of systems with detailed balance and with
the complex balance are equivalent. Let us consider a thermodynamic state given by the
vector of potentials $\check{\mu}$ defined by Equation (\ref{DualVariables}). Let all the
reactions in the reaction mechanism be reversible ({\em i.e.}, for every transition
$\Theta_i \to \Theta_j$ the reverse transition $\Theta_i \leftarrow \Theta_j$ is allowed
and the corresponding edge belongs to the digraph of complex transformations). Calculate
the right hand side of the kinetic equations (\ref{KinUrChemRev}) with the detailed
balance condition given by Equation (\ref{detailed balance}) for a given value of
$\check{\mu}$ and all possible values of $\varphi_{\rho}^+=\varphi_{\rho}^-$. The set of
these values of $\dot{N}$ is a convex cone. Denote this cone $\mathbf{Q}_{\rm
DB}(\check{\mu})$. For the same transition graph, calculate the right hand side of the
kinetic equation (\ref{KinUrChem}) under the complex balance condition
(\ref{complexbalanceGENKIN}). The set of these values of $\dot{N}$ is also a convex cone.
Denote it $\mathbf{Q}_{\rm CB}(\check{\mu})$. It is obvious that $\mathbf{Q}_{\rm
DB}(\check{\mu}) \subseteq \mathbf{Q}_{\rm CB}(\check{\mu})$. Surprisingly, these cones
coincide. In \cite{GorbanEq2012arX} we proved this fact on the basis of the
Michaelis-Menten-Stueckelberg theorem \cite{GorbanShahzad2011} about connection of the
macroscopic GMAL kinetics and the complex balance condition with the Markov microscopic
description and under some asymptotic assumptions. Below a direct proof is presented.
\begin{theorem}[Local equivalence of detailed and complex balance]\label{theor:equivalence}
\begin{equation}
\mathbf{Q}_{\rm DB}(\check{\mu}) = \mathbf{Q}_{\rm CB}(\check{\mu})
\end{equation}
\end{theorem}
\begin{proof}
Because of the cycle decomposition (Proposition~\ref{prop:Cycle decomposition}) it is
sufficient to prove this theorem for simple normalized cycles. Let us use induction on
the cycle length $k$. For $k=2$ the transition graph is $\Theta_1 \rightleftharpoons
\Theta_2$ and the detailed balance condition (\ref{detailed balance}) coincides with the
complex balance condition (\ref{complexbalanceGENKIN}). Assume that for the cycles of the
length below $k$ the theorem is proved. Consider a normalized simple cycle $\Theta_1 \to
\Theta_2 \to \ldots \Theta_k \to \Theta_1$, $\Theta_i = (y_i, A)$. The corresponding
kinetic equations are
\begin{equation}\label{CycleKinetics}
\begin{split}
\frac{\D N}{\D t}=&(y_2-y_1)\exp(\check{\mu},y_1)+ (y_3-y_2) \exp(\check{\mu},y_2)
+\ldots \\ &+ (y_{k-1}-y_k)\exp(\check{\mu},y_{k-1})+(y_{k}-y_1)\exp(\check{\mu},y_{k})
\end{split}
\end{equation}
At the equilibrium, all systems with detailed balance or with complex balance give
$\dot{N}=0$. Assume that the state $\check{\mu}$ is non-equilibrium and therefore not all
the Boltzmann factors $\exp(\check{\mu},y_i)$ are equal. Select $i$ such that the
Boltzmann factor $\exp(\check{\mu},y_i)$ has minimal value, while for the next position
in the cycle this factor becomes bigger. We can use a cyclic permutation and assume that
the factor $\exp(\check{\mu},y_1)$ is the minimal one and
$\exp(\check{\mu},y_2)>\exp(\check{\mu},y_1)$.

Let us find a kinetic factor $\varphi$ such that the reaction system consisting of two
cycles, a cycle of the length 2 with detailed balance $\Theta_1
\underset{\varphi}{\overset{\varphi}{\rightleftharpoons}} \Theta_2$ (here the kinetic
factors are shown above and below the arrows) and a simple normalized cycle of the length
$k-1$, $\Theta_2 \to \ldots \Theta_k \to \Theta_2$, gives the same $\dot{N}$ at the state
$\check{\mu}$ as the initial scheme. We obtain from  Equation (\ref{CycleKinetics}) the
following necessary and sufficient condition
$$(y_1-y_k)\exp(\check{\mu},y_k)+(y_2-y_1)\exp(\check{\mu},y_1)=(y_2-y_k)\exp(\check{\mu},y_k)+\varphi
(y_2-y_1)(\exp(\check{\mu},y_1)-\exp(\check{\mu},y_2))$$. It is sufficient to equate here
the coefficients at every $y_i$ ($i=1,2,k$). The result is
$$\varphi=\frac{\exp(\check{\mu},y_k)-\exp(\check{\mu},y_1)}{\exp(\check{\mu},y_2)-\exp(\check{\mu},y_1)}$$
By the induction assumption we proved that theorem for the cycles of arbitrary length
and, therefore, it is valid for all reaction schemes with complex balance.
\end{proof}

The cone $\mathbf{Q}_{\rm DB}(\check{\mu})$ of the possible values of $\dot{N}$ in
Equation (\ref{KinUrChemRev}) is a polyhedral cone with finite set of extreme rays at any
non-equilibrium state $\check{\mu}$ for the systems with detailed balance. Each of its
extreme rays has the direction vector of the form
\begin{equation}
\gamma_{\rho}\mbox{sign}(\exp(\check{\mu},\alpha_{\rho})-\exp(\check{\mu},\beta_{\rho}))
\end{equation}
This follows from the form of the reaction rate presented by Equation
(\ref{DBreactionrate}) for the kinetic equations (\ref{KinUrChemRev}). Following
Theorem~\ref{theor:equivalence}, the cone of the possible values of $\dot{N}$ for systems
with complex balance has the same set of extreme rays. Each extreme ray corresponds to a
single reversible elementary reaction with the detailed balance condition (\ref{detailed
balance}).

\subsection{General $H$-Theorem for GMAL}

Consider GMAL kinetics with the given reaction mechanism presented by stoichiometric
equations (\ref{reversibleMechanism}) and the detailed balance condition (\ref{detailed
balance}). The reaction rates of the elementary reaction for the kinetic equations
(\ref{KinUrChemRev}) are proportional to the nonnegative parameter $\varphi_{\rho}$ in
Equation (\ref{DBreactionrate}). These $m$ nonnegative numbers $\varphi_{\rho}$
($\rho=1,\ldots, m$) are independent in the following sense: for any set of values
$\varphi_{\rho} \geq 0$ the kinetic equations~(\ref{KinUrChemRev}) satisfy the
$H$-theorem in the form of Equation (\ref{DissTanh}):
\begin{equation}\label{DissTanhPhi}
\frac{\D H}{\D t}=-V \sum_{\rho} \varphi_{\rho}(\exp(\alpha_{\rho},\check{\mu})
+ \exp(\beta_{\rho},\check{\mu}))\mathbb{A}_{\rho} \tanh \frac{\mathbb{A}_{\rho}}{2} \leq 0
\end{equation}
Therefore, nonnegativity is the only a priori restriction on the values of
$\varphi_{\rho}$ ($\rho=1,\ldots, m$).

One Lyapunov function for the GMAL kinetics with the given reaction mechanism and the
detailed balance condition obviously exists. This is the thermodynamic Lyapunov function
$H$ used in GMAL construction. For ideal systems (in particular, for master equation) $H$
has the standard form $\sum_i N_i(\ln(c_i/c_i^*)-1)$ given by Equation
(\ref{PerfectFreeEntropy}). Usually, $H$ is assumed to be convex and some singularities
(like $c\ln c$) near zeros of $c$ may be required for positivity preservation in kinetics
($\dot{N}_i\geq 0$ if $c_i=0$). The choice of the thermodynamic Lyapunov function for
GMAL construction is wide. We consider kinetic equations in a compact convex set $U$ and
assume $H$ to be convex and continuous in $U$ and differentiable in the relative interior
of $U$ with derivatives continued by continuity to $U$.

Assume that we select the thermodynamic Lyapunov function $H$ and the reaction mechanism
in the form (\ref{reversibleMechanism}). Are there other universal Lyapunov functions for
GMAL kinetics with detailed balance and given mechanism? ``Universal'' here means
``independent of the choice of the nonnegative kinetic factors''.

For a given reaction mechanism we introduce the partial equilibria criterion by analogy
to Definition~\ref{def:PatEqCrit}. Roughly speaking, a convex function $F$ satisfies this
criterion if its conditional minima correctly describe the partial equilibria of
elementary reactions.

For each elementary reaction $\sum_i\alpha_{\rho i}A_i \rightleftharpoons \sum_i
\beta_{\rho i} A_i$ from the reaction mechanism given by the stoichiometric equations
(\ref{reversibleMechanism}) and any $X\in U$ we define an interval of a straight line
\begin{equation}
I_{X, \rho}= \{X+\lambda \gamma_{\rho} \, |
\, \lambda \in \mathbb{R}\} \cap U.
\end{equation}

\begin{definition}[Partial equilibria criterion for GMAL]\label{def:partEquilGMAL}A convex function $F(N)$ on $U$ satisfies the partial
equilibria criterion with a given thermodynamic Lyapunov function $H$ and reversible
reaction mechanism given by stoichiometric equations (\ref{reversibleMechanism}) if
\begin{equation}
\underset{{N\in I_{X, \rho}}}{\operatorname{argmin}} H(N) \subseteq \underset{{N\in I_{X, \rho}}}{\operatorname{argmin}} F(N)
\end{equation}
for all $X\in U$, $\rho=1,\ldots ,m$.
\end{definition}

\begin{theorem}\label{theorem:DetBalGenHthGMAL}A convex function $F(N)$ on $U$ is a Lyapunov function for all kinetic equations (\ref{KinUrChemRev})
with the given thermodynamic Lyapunov function $H$ and reaction rates presented by
Equation (\ref{DBreactionrate}) (detailed balance) if and only if it satisfies the
partial equilibria criterion (Definition~\ref{def:partEquilGMAL}).
\end{theorem}
\begin{proof}
The partial equilibria criterion is necessary because $F(N)$ should be a Lyapunov
function for a reaction mechanism that consists of any single reversible reaction from
the reaction mechanism (\ref{reversibleMechanism}). It is also sufficient because for the
whole reaction mechanism the kinetic equations (\ref{KinUrChemRev}) are the conic
combinations of the kinetic equations for single reversible reactions from the reaction
mechanism (\ref{reversibleMechanism}).
\end{proof}

For the general reaction systems with complex balance we can use the theorem about local
equivalence (Theorem~\ref{theor:equivalence}). Consider a GMAL reaction system with the
mechanism (\ref{stoichiometricequation}) and the complex balance condition.

\begin{theorem}\label{theor:ComplBalGMALH}A convex function $F(N)$ on U is a Lyapunov function for all kinetic equations
(\ref{KinUrChem}) with the given thermodynamic Lyapunov function $H$ and the complex
balance condition (\ref{complexbalanceGENKIN}) if it satisfies the partial equilibria
criterion (Definition~\ref{def:partEquilGMAL}).
\end{theorem}
\begin{proof}
The theorem follows immediately from Theorem~\ref{theorem:DetBalGenHthGMAL} about
Lyapunov functions for systems with detailed balance and the theorem about local
equivalence between systems with local and complex balance
(Theorem~\ref{theor:equivalence}).
\end{proof}

The general $H$-theorems for GMAL is similar to Theorem~\ref{theor:MarkH} for Markov
chains. Nevertheless, many non-classical universal Lyapunov functions are known for
master equations, for example, the $f$-divergences given by Equation (\ref{Morimoto}),
while for a nonlinear reaction mechanism it is difficult to present a single example
different from the thermodynamic Lyapunov function or its monotonic transformations. The
following family of example generalizes Equation (\ref{NonclassicExampleEnt}).
\begin{equation}
F(N)=H(N)+\varepsilon f(N)\prod_{\rho} (\exp(\alpha_{\rho},\check{\mu})-\exp(\beta_{\rho},\check{\mu}))^2
\end{equation}
where $f(N)$ is a non-negative differentiable function and $\varepsilon>0$ is a
sufficiently small number. This function satisfies the conditional equilibria criterion.
For continuous $H(N)$ on compact $U$ with the spectrum of the Hessian uniformly separated
from zero, this $F(N)$ is convex for sufficiently small $\varepsilon>0$.

\section{Generalization: Weakened Convexity Condition, Directional Convexity and Quasiconvexity}

In all versions of the general $H$-theorems we use convexity of the Lyapunov functions.
Strong convexity of the thermodynamic Lyapunov functions $H$ (or even positive
definiteness of its Hessian) is needed, indeed, to provide reversibility of the Legendre
transform $N\leftrightarrow \nabla H$.

\begin{figure}[t]
\centering{
\includegraphics[height=0.35\textwidth]{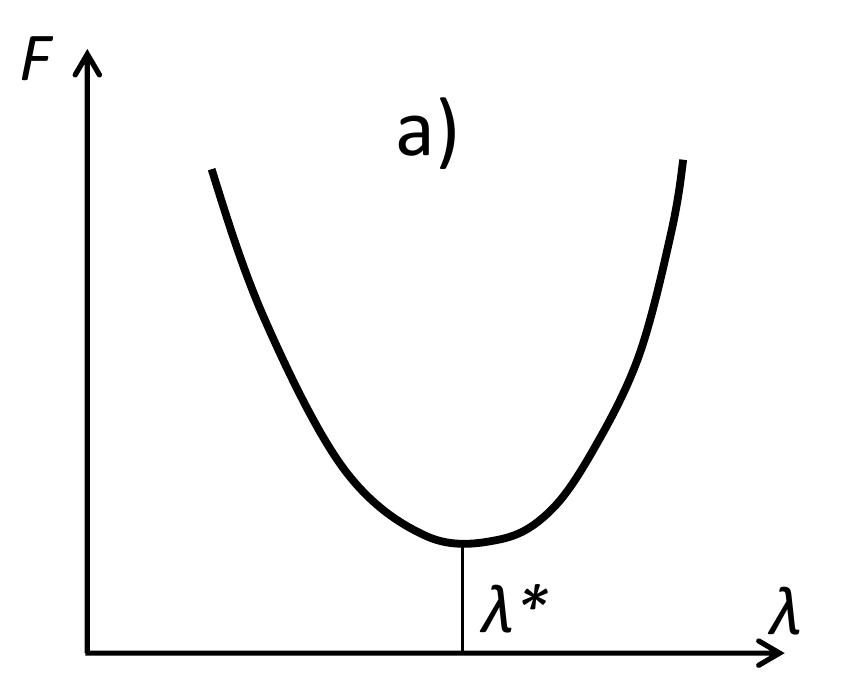} \hspace{5mm}
\includegraphics[height=0.35\textwidth]{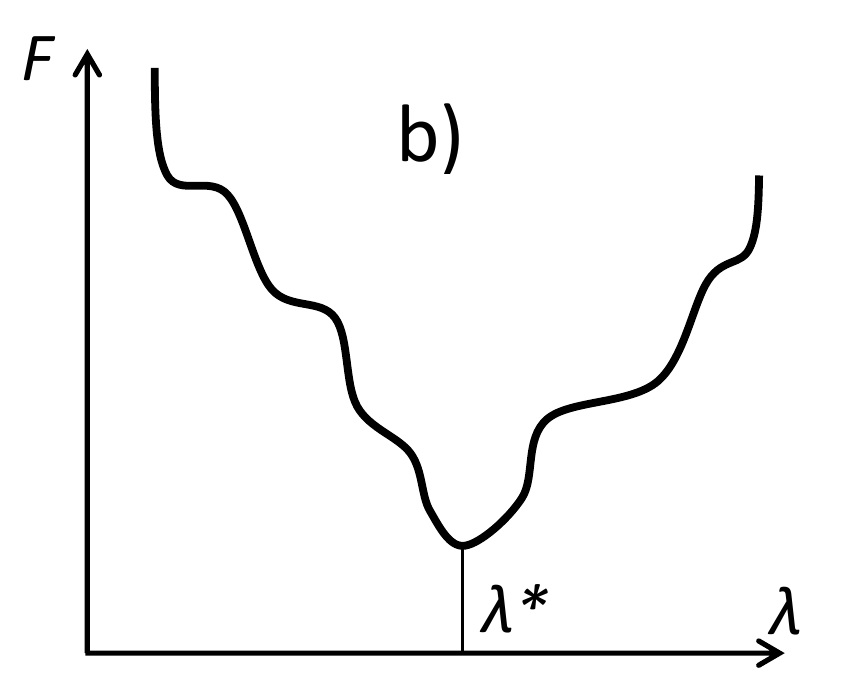}
\caption{Monotonicity on both sides of the minimizer $\lambda^*$ for convex (a) and non-convex but quasiconvex (b) functions \label{fig:ConvexQconv}}}
\end{figure}

For the kinetic Lyapunov functions that satisfy the partial equilibria criterion, we use,
actually, a rather weak consequence of convexity in restrictions on the straight lines
$X+\lambda \gamma$ (where $\lambda \in \mathbb{R}$ is a coordinate on the real line,
$\gamma $ is a stoichiometric vector of an elementary reaction): if
$\lambda^*=\underset{\lambda \in \mathbb{R}}{\operatorname{argmin}} F(X+\lambda \gamma)$
then on the half-lines (rays) $\lambda \geq \lambda^*$ and $\lambda \leq \lambda^*$
function $F(X+\lambda \gamma)$ is monotonic. It does not decrease for $\lambda \geq
\lambda^*$ and does not increase for $\lambda \leq \lambda^*$. Of course, convexity is
sufficient (Figure~\ref{fig:ConvexQconv}a) but a much weaker property is needed
(Figure~\ref{fig:ConvexQconv}b).

A function $F$ on a convex set $U$ is {\em quasiconvex} \cite{Greenberg1971} if all its
sublevel sets are convex. It means that for every $X,Y \in U$
\begin{equation}
F(\lambda X+ (1-\lambda)Y) \leq \max\{F(X), F(Y)\} \mbox{ for all } \lambda \in [0,1]
\end{equation}
In particular, a function $F$ on a segment is quasiconvex if all its sublevel sets are
segments.

Among many other types of convexity and quasiconvexity (see, for example
\cite{Ponstein1967}) two are important for the general $H$-theorem. We do not need
convexity of functions along all straight lines in $U$. It is sufficient that the
function is convex on the straight lines $X+\mathbb{R} \gamma_{\rho}$, where
$\gamma_{\rho}$ are the stoichiometric (direction) vectors of the elementary reactions.

Let $D$ be a set of vectors. A function $F$ is $D$-convex if its restriction to each line
parallel to a nonzero $v\in D$ is convex \cite{Matousek2001}. In our case, $D$ is the set
of stoichiometric vectors of the transitions, $D=\{\gamma_{\rho} \, |\, \rho=1,\ldots ,
m\}$. We can use this {\em directional convexity} instead of convexity in
\mbox{Propositions~\ref{prop:DBMarkH}, \ref{prop:CBMarkH}} and
Theorems~\ref{theor:MarkH}, \ref{theorem:DetBalGenHthGMAL}, \ref{theor:ComplBalGMALH}.

Finally, we can relax the convexity conditions even more and postulate {\em directional
quasiconvexity}~\cite{Hwang1996} for the set of directions $D=\{\gamma_{\rho} \, |\,
\rho=1,\ldots , m\}$. Propositions~\ref{prop:DBMarkH}, \ref{prop:CBMarkH} and
Theorems~\ref{theor:MarkH}, \ref{theorem:DetBalGenHthGMAL}, \ref{theor:ComplBalGMALH}
will be still true if the functions are continuous, quasiconvex in restrictions on all
lines $X+\mathbb{R} \gamma_{\rho}$ and satisfy the partial equilibria criterion.

Relations between these types of convexity are schematically illustrated in
Figure~\ref{Convexities}.

\begin{figure}[t]
\centering{
\includegraphics[height=0.35\textwidth]{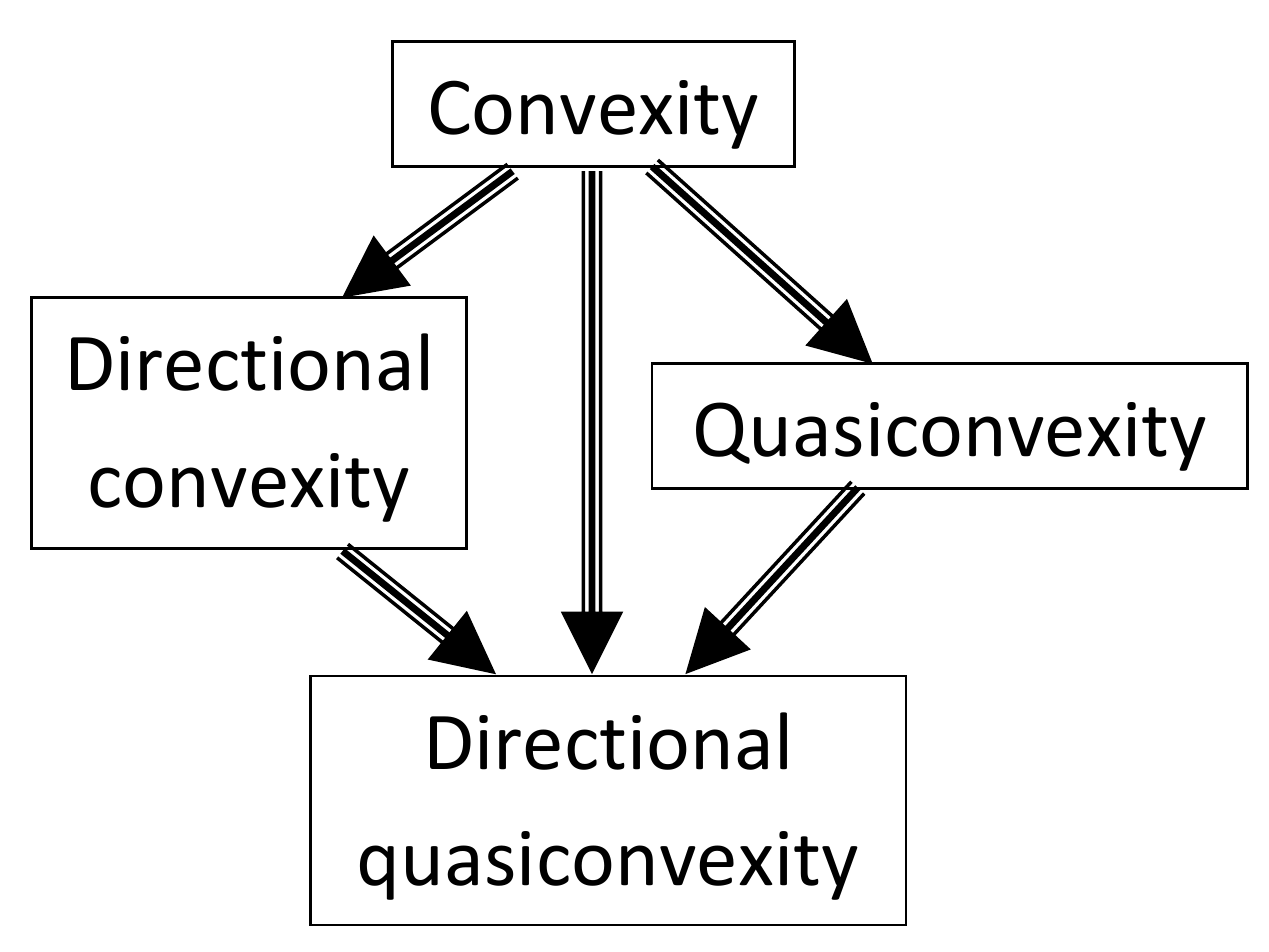}
\caption{Relations between different types of convexity \label{Convexities}}}
\end{figure}

\section{Discussion}

Many non-classical entropies are invented and applied to various problems in physics and
data analysis. In this paper, the general necessary and sufficient criterion for the
existence of $H$-theorem is proved. It has a simple and physically transparent form: the
convex divergence (relative entropy) should properly describe the partial equilibria for
transitions $A_i \rightleftharpoons A_j$. It is straightforward to check this partial
equilibria criterion. The applicability of this criterion does not depend on the detailed
balance condition and it is valid both for the class of the systems with detailed balance
and for the general first order kinetics without this assumption.

If an entropy has no $H$-theorem (that is, it violates the second law and the data
processing lemma) then there should be unprecedentedly strong reasons for its use.
Without such strong reasons we cannot employ it. Now, I cannot find an example of
sufficiently strong reasons but people use these entropies in data analysis and we have
to presume that they may have some hidden reasons and that these reasons may be
sufficiently strong. We demonstrate that this problem arises even for such popular
divergences like Euclidean distance or Itakura-Saito distance.

The general $H$-theorem is simply a reduction of a dynamical question (Lyapunov
functionals) to a static one (partial equilibria). It is not surprising that it can be
also proved for nonlinear Generalized Mass Action Law kinetics. Here kinetic systems with
complex balance play the role of the general Markov chains, whereas the systems with
detailed balance correspond to the reversible Markov chains. The requirement of convexity
of Lyapunov functions can be relaxed to the directional convexity (in the directions of
reactions) or even directional quasiconvexity.

For the reversible Markov chains presented by Equations (\ref{QuasiChemKol}) with the
classical entropy production formula (\ref{entropyprodDelBal}), every universal Lyapunov
function $H$ should satisfy inequalities
\begin{equation}
\left(\frac{p_j}{p_j^*}-\frac{p_i}{p_i^*}\right)
\left(\frac{\partial H(P)}{\partial p_j}-\frac{\partial H(P)}{\partial p_i}\right)\leq 0
\mbox{ for all } i, j, i\neq j
\end{equation}
These inequalities are closely related to another generalization of convexity, the Schur
convexity \cite{MarshallOlkin2011}. They turn into the definition of the Schur convexity
when equilibrium is the equidistribution with $p_i^*=1/n$ for all $i$. Universal Lyapunov
functions for nonlinear kinetics give one more generalization of the Schur convexity.

Introduction of many non-classical entropies leads to the ``uncertainty of uncertainty''
phenomenon: we measure uncertainty by entropy but we have uncertainty in the entropy
choice \cite{GorbanCAMWA2013}. The selection of the appropriate entropy and introduction
of new entropies are essentially connected with the class of kinetics. $H$-theorems in
physics are formalizations of the second law of thermodynamics: entropy of isolated
systems should increase in relaxation to equilibrium. If we know the class of kinetic
equations (for example, the Markov kinetics given by master equations) then the $H$
theorem states that it is possible to use this entropy with the given kinetics. If we
know the entropy and are looking for kinetic equations then such a statement turns into
the thermodynamic restriction on the thermodynamically admissible kinetic equations. For
information processing, the class of kinetic equations describes possible manipulations
with data. In this case, the $H$-theorems mean that under given class of manipulation the
information does not increase. It is not possible to compare different entropies without
any relation to kinetics. It is useful to specify the class of kinetic equations, for
which they are the Lyapunov functionals. For the GMAL equations, we can introduce the
dynamic equivalence between divergences (free entropies or conditional entropies). Two
functionals $H(N)$ and $F(N)$ in a convex set $U$ are {\em dynamically consistent} with
respect to the set of stoichiometric vectors $\{\gamma_{\rho}\}$ ($\rho=1, \ldots , m$)
if

\begin{enumerate}
\item[(1)]{$F$ and $H$ are directionally quasiconvex functions in directions
    $\{\gamma_{\rho}\}$ ($\rho=1, \ldots , m$)}
\item[(2)]{For all $\rho=1, \ldots , m$ and $N\in U$
$$ (\nabla_N F(N),\gamma_{\rho})(\nabla_N H(N),\gamma_{\rho})\geq 0$$}
\end{enumerate}
For the Markov kinetics, the partial equilibria criterion is sufficient for a convex
function $H(P)$ to be dynamically consistent with the relative entropy $\sum_i p_i
(\ln(p_i/p_i^*)-1)$ in the unit simplex $\Delta_n$. For GMAL, any convex function $H(N)$
defines a class of kinetic equations. Every reaction mechanism defines a family of
kinetic equations from this class and a class of Lyapunov functions $F$, which are
dynamically consistent with $H$. The main message of this paper is that it is necessary
to discuss the choice of the non-classical entropies in the context of kinetic equations.

\section*{Appendix: Quasiequilibrium entropies and forward--invariant peeling}

\subsection*{A1. Maximum of quasiequilibrium entropies -- a new family \\ of universal Lyapunov functions for generalized mass action law}

The general $H$ theorems for the Generalized Mass Action Law (GMAL) and for its linear
version, master equation, look very similar. For the linear systems many Lyapunov
functionals are known in the explicit form: for every convex function $h$ on the positive
ray $\mathbb{R}_+$ we have such a functional (\ref{Morimoto}). On the contrary, for the
nonlinear systems we, typically, know the only Lyapunov function $H$, it is the
thermodynamic potential which is used for the system construction. The situation looks
rather intriguing and challenging: for every finite reaction mechanism there should be
many Lyapunov functionals, but we cannot construct them. (There is no chance to find many
Lyapunov functions for {\em all} nonlinear mechanisms together under given thermodynamics
because in this case the cone of the possible velocities $\dot{N}$ is a half-space and
locally there is the only divergence with a given tangent hyperplane. Globally, such a
divergence can be given by an arbitrary monotonic function on the thermodynamic tree
\cite{G11984,GorbanSIADS2013}).

In this Appendix, we present a general procedure for the  construction of a family of new
Lyapunov functionals from $H$ for nonlinear GMAL kinetics and a given reaction mechanism.
We will use two auxiliary construction, the quasiequilibrium entropies (or divergences)
and the forward--invariant peeling.

Let us consider isochoric systems (constant volume $V$). For them, concentrations $c_i$
(intensive variables) and amounts $N_i$ (extensive variables are proportional with a
constant extensive factor $V$ and we take $N_i=c_i$ in a standard unit volume without
loss of generality.

We assume that $H$ is strongly  convex in the second approximation in $\mathbb{R}_+^n$.
This means that it is twice differentiable and the Hessian $\partial^2 H/\partial N_i
\partial N_j$ is positively definite in $\mathbb{R}_+^n$. In addition, we assume
logarithmic singularities of the partial derivatives of $H$ near zeros of concentrations:

\begin{equation}\label{FreeEnLogSIng}
H(N)=\sum_i N_i (\ln c_i-1 + \mu_{0i}(c))\, ,
\end{equation}
where the functions $\mu_{0i}(c)$ are bounded continuously differentiable functions in a
vicinity of the non-negative orthant. This assumption corresponds to the physical
hypothesis about the logarithmic singularity of the chemical potentials, $\mu_i= RT\ln
c_i +\ldots$ where $\ldots$ stands for a continuous function of $c,T$, and to the
supposition about the classical mass action law for small concentrations. Assume also
that all the described properties of $H$ hold for its restrictions on the faces of
$\overline{\mathbb{R}_+^n}$: these restrictions are strictly convex, differentiable in
the relative interior, etc.

For every linear subspace $E\subset \mathbb{R}^n$ and a given composition vector $N^0 \in
\mathbb{R}_+^n$ the {\em quasiequilibrium composition} is the partial equilibrium
$$N^*_{E}(N^0)= \underset{{N\in
(N^0+E)\cap \mathbb{R}_+^n}}{\operatorname{argmin}}H(N)$$

The {\em quasiequilibrium divergence} is the value of $H$ at the partial equilibrium:
$$H^*_{E}(N^0) = \min_{N\in (N^0+E)\cap \mathbb{R}_+^n} H(N)$$

Due to the assumption about strong convexity of $H$ and logarithmic singularity
(\ref{FreeEnLogSIng}), for a positive vector $N^0 \in \mathbb{R}_+^n$ and a subspace $E
\subset \mathbb{R}^n$ the quasiequilibrium composition $N^*_{E}(N^0)$ is also positive.

Such quasiequilibrium ``entropies'' are discussed by Jaynes \cite{Jaynes1965}. He
considered the quasiequilibrium $H$-function as the Boltzmann $H$-function $H_{\rm B}$ in
contrast to the original Gibbs $H$-function, $H_{\rm G}$. The Gibbs $H$-function is
defined for the distributions on the phase space of the mechanical systems. The Boltzmann
function is a conditional minimum of the Gibbs function, therefore the inequality holds
$H_{\rm B} \leq H_{\rm G}$ \cite{Jaynes1965}. Analogously,
$$H^*_{E}(N^0)\leq H(N^0)$$ and this inequality turns into the equality if and only if
$N^0$ is the quasiequilibrium state for the subspace $E$: $N^0=N^*_{E}(N^0)$. After
Jaynes, these functions are intensively used in the discussion of time arrow
\cite{Ocherki1986,Lebowith1993,GoldsteinLebow2004}. In the theory of information,
quasiequilibrium was studied in detail under the name {\em information projection} (or
I-projection) \cite{CsiszarMatus2003}.

Analysis of partial equilibria is useful in chemical engineering in the presence of
uncertainty: when the reaction rate constants are unknown then the chains of partial
equilibria together with information about the thermodynamically preferable directions of
reactions may give some important information about the process \cite{GorbanKagan2006}.

Let us prove several elementary properties of $H^*_{E}(N)$. Let $E$ and $L$ be subspaces
of $\mathbb{R}^n$.
\begin{proposition}
\begin{enumerate}
\item The function $H^*_{E}(N)$ is convex.
\item If $E$ is a proper subspace of $R^n$ then the function $H^*_{E}(N)$ is not
    strictly convex: for each $N \in \mathbb{R}_+^n$ the level set $\{N' \, | \,
    H^*_{E}(N')=H^*_{E}(N)\}$ includes faces $(N+E)\cap \mathbb{R}_+^n$.
\item The function $H^*_{E}([N])$ is strictly convex on the quotient space
    $\mathbb{R}_+^n/E$ (here, $[N] \in \mathbb{R}_+^n/E$ is the equivalence class,
    $[N]=(N+E)\cap \mathbb{R}_+^n$).
\item If  $E \subseteq L$ then $H^*_{E}(N) \geq H^*_{L}(N)$ and this inequality turns
    into the equality if and only if the corresponding quasiequlibria coincide:
    $N^*_{E}(N)=N^*_{L}(N)$ (this is a generalization of the Jaynes inequality
    $H_{\rm B} \leq H_{\rm G}$).
\item If $N=N^*_E(N)$ then  $H^*_E(N) \geq H^*_L(N)$ for all $L$ and this inequality
    turns into the equality if and only if $N=N^*_{(E+L)}(N)$.
\end{enumerate}
\end{proposition}
\begin{proof}
\begin{enumerate}
\item{Convexity of $H^*_{E}(N)$ means that for every positive $N^1$ and $N^2$ and a
    number $\lambda \in [0,1]$ the inequality holds: $$H^*_{E}(\lambda N^1+
    (1-\lambda)N^2)\leq \lambda H^*_{E}(N^1)+(1-\lambda)H^*_{E}(N^2)$$ Let us prove
    this inequality. First, $$H(\lambda N^*_E (N^1)+(1-\lambda) N^*_E (N^2) \leq
    \lambda H(N^*_{E}(N^1))+(1-\lambda)H(N^*_{E}(N^2))$$ because convexity $H$.
    Secondly, $H(N^*_{E}(N^{1,2}))=H^*_{E}(N^{1,2})$ by definition and the last
    inequality reads $$H(\lambda N^*_E (N^1)+(1-\lambda) N^*_E (N^2) \leq \lambda
    H^*_{E}(N^1)+(1-\lambda)H^*_{E}(N^2)$$ Finally, $N^*_{E}(N^{1,2}) \in N^{1,2}+E$,
    hence, $$\lambda N^*_E (N^1)+(1-\lambda) N^*_E (N^2)\in \lambda N^1+(1-\lambda)
    N^2 + E$$ and $H(\lambda N^*_E (N^1)+(1-\lambda) N^*_E (N^2))\geq H^*_{E}(\lambda
    N^1+ (1-\lambda)N^2)$ because the last value is the minimum of $H$ on the linear
    manifold $\lambda N^1+(1-\lambda) N^2 + E$. Inequality is proven.}
\item{Indeed, the function $H^*_E$ is constant on the set $(N+E)\cap \mathbb{R}_+^n$,
    by construction.}
\item{In the proof of item 1 the inequality $H(\lambda N^*_E (N^1)+(1-\lambda) N^*_E
    (N^2))>  H^*_{E}(\lambda N^1+ (1-\lambda)N^2)$ is strong for $\lambda \neq 0,1$
    and $N^1-N^2\notin E$. Therefore, under these conditions the convexity inequality
    is strong.}
\item{If $E \subseteq L$ then $N+E\subset N+L$ and $H^*_{E}(N) \geq H^*_{L}(N)$ by
    definition of $H^*$ as a conditional minimum. This inequality turns into the
    equality if and only if the corresponding quasiequlibria coincide because of
    strong convexity of $H$.}
\item{This follows directly from the definitions of $H^*_E(N)$ as a conditional
    minimum and $N^*_E(N)$ as the corresponding minimizer of $H$ on $N+E \cap
    \mathbb{R}^n_+$}
\end{enumerate}
\end{proof}

Consider the reversible reaction mechanism (\ref{reversibleMechanism}) with the set of
the stoichiometric vectors $\Upsilon$. For each $\Gamma \subset \Upsilon$ we can take
$E={\rm Span}(\Gamma)$ and define the quasiequilibrium. The subspace ${\rm Span}(\Gamma)$
may coincide for different $\Gamma$ and the quasiequilibrium depends on the subspace $E$
only, therefore, it is useful to introduce the set of these subspaces for a given
reaction mechanism (\ref{reversibleMechanism}). Let $\mathcal{E}_{\Upsilon}$ be the set
of all subspaces of the form $E={\rm Span}(\Gamma)$ ($\Gamma \subset \Upsilon$). For each
dimension $k$ we denote $\mathcal{E}^k_{\Upsilon}$ the set of $k$-dimensional subspaces
from $\mathcal{E}_{\Upsilon}$.

For each dimension $k=0,\ldots, {\rm rank} (\Upsilon)$ we define the function $H^{k,
\max}_{\Upsilon}$: $H^{0, \max}_{\Upsilon}=H$, and for $0<k \leq {\rm rank} (\Upsilon)$
\begin{equation}\label{Hmax}
H^{k, \max}_{\Upsilon}(N) =\max_{E \in \mathcal{E}^k_{\Upsilon}} H^*_E(N)
\end{equation}

Immediate consequence of the definition of the quasiequilibrium divergence and
Theorem~\ref{theorem:DetBalGenHthGMAL} is:
\begin{proposition}\label{Prop:HmaxDetBal}
$H^{1,\max}_{\Upsilon}(N)$ is a Lyapunov function in $\mathbb{R}_+^n$ for all kinetic
equations (\ref{KinUrChemRev}) with the given thermodynamic Lyapunov function $H$,
reaction rates presented by Equation (\ref{DBreactionrate}) (detailed balance) and the
reversible reaction mechanism with the set of stoichiometric vectors $\Gamma \subseteq
\Upsilon$.
\end{proposition}
\begin{proof}
$H^{1, \max}_{\Upsilon}(N)$ is a convex function as the maximum of several convex
functions. Let us consider a restriction of this function onto an interval of the
straight line $I=(N^0+\mathbb{R} \gamma)\cap \mathbb{R}^n_+$ for a stoichiometric vector
$\gamma \in \Upsilon$. The partial equilibrium $N^{**}=N^*_{\{\mathbb{R} \gamma\}}(N^0)$
is the minimizer of $H(N)$ on $I$. Assume that this partial equilibrium is not a partial
equilibrium for other 1D subspaces $E\in \mathcal{E}^1_{\Upsilon}$. Then for all $E\in
\mathcal{E}^1_{\Upsilon}$ ($E\neq \mathbb{R} \gamma$) $H^*_E(N^{**})<H^*_{\{\mathbb{R}
\gamma\}}(N^0)$ and
$$H^{1,\max}_{\Upsilon}(N)=H^*_{\{\mathbb{R} \gamma\}}(N)$$ in some vicinity of
$N^{**}$. This function is constant on $I$.

If a convex function $h$ on an interval $I$ is constant on an subinterval $J=(a,b)
\subset I$ ($a \neq b$) then the value $h(J)$ is the minimum of $h$ on $I$.  Therefore,
$N^*_{\{\mathbb{R} \gamma\}}(N^0)$ is a minimizer of the convex function $H^{1,
\max}_{\Upsilon}(N)$ on $I$ in the case, when $N^{**}$ is not a partial equilibrium for
other 1D subspaces $E\in \mathcal{E}^1_{\Upsilon}$ ($E\neq \mathbb{R} \gamma$).

Let us assume now that the partial equilibrium $N^{**}=N^*_{\{\mathbb{R} \gamma\}}(N^0)$
is, at the same time, the partial equilibrium for several other $E\in
\mathcal{E}^1_{\Upsilon}$. Let $\mathcal{B}$ be the set of subspaces  $E\in
\mathcal{E}^1_{\Upsilon}$ for which $N^{**}$ is a partial equilibrium, i.e.
$H(N^{**})=H^*_E(N^{**})$. In this case, for all $E \notin \mathcal{B}$ ($E\in
\mathcal{E}^1_{\Upsilon}$) $H^*_E(N^{**})<H(N^{**})$. Therefore, in a sufficiently small
vicinity of $N^{**}$ the function $H^{1,\max}_{\Upsilon}(N)$ can be defined as
$$H^{1,\max}_{\Upsilon}(N)=\max_{E\in \mathcal{B}}H^*_E(N) $$
Point $N^{**}$ is a minimizer of $H$ on a linear manifold $N^{**}+(\bigoplus_{E\in B}
E)$. It is also a minimizer of convex function $H^{1,\max}_{\Upsilon}(N)$ on this linear
manifold, particularly, it is a minimizer of this convex function on the interval $I$.
(Convexity plays a crucial role in this reasoning because for convex functions the local
minima are the global ones.)

We proved that the function $H^{1, \max}_{\Upsilon}(N)$ satisfies the partial equilibria
criterion and, hence, it is a Lyapunov function in $\mathbb{R}_+^n$ for all kinetic
equations (\ref{KinUrChemRev}) with the given thermodynamic Lyapunov function $H$,
reaction rates presented by Equation (\ref{DBreactionrate}) (detailed balance) and the
reversible reaction mechanism with the set of stoichiometric vectors $\Gamma \subseteq
\Upsilon$.
\end{proof}
Let a positive vector $N^{**}$ be a minimizer of $H$ on  $(N^{**}+E)\cap \mathbb{R}^n$,
where $E$ is a linear subspace of $\mathbb{R}^n$. It may be useful to represent the
structure of the Lyapunov function $H^{1,\max}_{\Upsilon}(N)$ near $N^{**}$ in the
quadratic approximation. Assume that $H(N)$ is $m$ times continuously differentiable in
$\mathbb{R}^n$ for sufficiently large $m$. In the vicinity of $N^{**}$
$$H(N)-H(N^{**})=(DH)_{**}(\Delta)+ \frac{1}{2}\langle \Delta, \Delta \rangle_{**}+o(\|\Delta\|^2)$$
where $\Delta= N-N^{**}$, $(DH)_{**}=(DH)\left|_{N^{**}}\right.$ is the differential of
$H$ at $N^{**}$, and $\langle \bullet,\bullet \rangle_{**}=(\bullet, (D^2 H)_{**}
\bullet)$ is the {\em entropic inner product}, with the positive symmetric operator $(D^2
H)_{**}=(D^2 H)\left|_{N^{**}}\right.$ (the second differential of $H$ at $N^{**}$). The
entropic inner product is widely used in kinetics and nonequilibrium thermodynamics, see,
for example \cite{G11984,Ocherki1986,UNIMOLD,GorGorKar2004,InChLANL,GorKarLNP2005}.

Let us split $\mathbb{R}^n$ into the orthogonal sum: $\mathbb{R}^n=E\oplus E^{\bot}$,
$E^{\bot}$ is the orthogonal supplement to $E$ in the entropic inner product $\langle
\bullet,\bullet \rangle_{**}$. Each vector $\Delta\in \mathbb{R}^n$ is represented in the
form $\Delta=\Delta^{\|}\oplus \Delta^{\bot}$, where $\Delta^{\|} \in E$ and
$\Delta^{\bot} \in E^{\bot}$. By the definition of the partial equilibrium as a
conditional minimizer of $H$, $(DH)_{**}(\Delta^{\|})=0$, and we have the following
representation of $H$
$$H(N)-H(N^{**})=(DH)_{**}(\Delta^{\bot})+ \frac{1}{2}\langle \Delta^{\|}, \Delta^{\|} \rangle_{**}
+\frac{1}{2}\langle \Delta^{\bot}, \Delta^{\bot} \rangle_{**}  + o(\|\Delta\|^2)$$ In particular, when $\Delta \in E$ ($\Delta^{\bot}=0$), this formula gives
$$H(N)-H(N^{**})=\frac{1}{2}\langle \Delta, \Delta \rangle_{**} + o(\|\Delta\|^2)$$

From these formulas, we easily get  the approximations of $N^*_E(N)$ and $H^*_E(N)$ in a
vicinity of $N^{**}$. Let $N-N^{**}=\Delta=\Delta^{\|}\oplus \Delta^{\bot}$. Then
\begin{equation}
N^*_E(N)-N^{**}=\Delta^{\bot}+o(\|\Delta\|)
\end{equation}
 in particular, $(N^*_E(N)-N^{**})^{\bot}=\Delta^{\bot}$ (exactly) and
 $(N^*_E(N)-N^{**})^{\|}=o(\|\Delta\|)$. Therefore,
\begin{equation}
H^*_E(N)-H(N^{**})=H(N^*_E(N))-H(N^{**})=(DH)_{**}(\Delta^{\bot})+\frac{1}{2}\langle
\Delta^{\bot}, \Delta^{\bot} \rangle_{**}  + o(\|\Delta\|^2)
\end{equation}
If $E$ is a 1D subspace with the directional vector $\gamma$ ($E=\mathbb{R} \gamma$) then
$$\Delta^{\|}=\frac{\gamma \langle \gamma, \Delta  \rangle_{**}}{\langle \gamma, \gamma  \rangle_{**}}, \;\;
\Delta^{\bot} = \Delta - \frac{\gamma \langle \gamma, \Delta  \rangle_{**}}{\langle
\gamma, \gamma  \rangle_{**}}$$
 here, $\frac{\gamma \langle \gamma|}{\langle \gamma, \gamma  \rangle}$ is the
 orthogonal projector onto $E$ and $1-\frac{\gamma \langle \gamma|}{\langle \gamma, \gamma
 \rangle}$ is the orthogonal projector onto $E^{\bot}$, the orthogonal complement to $E$.

Let us use the normalized vectors $\gamma$. In this case, $$N^*_E(N)-N^{**}=\Delta -
\gamma \langle \gamma, \Delta \rangle_{**}+o(\|\Delta\|)$$
\begin{equation}\label{QuadraticH*}
H^*_E(N)-H(N^{**})=(DH)_{**}(\Delta - \gamma \langle \gamma, \Delta
\rangle_{**})+\frac{1}{2}\langle \Delta - \gamma \langle \gamma, \Delta
\rangle_{**},\Delta - \gamma \langle \gamma, \Delta \rangle_{**} \rangle_{**} +
o(\|\Delta\|^2)
\end{equation}

Let us assume now that $N^{**}$ is the partial equilibrium  for several (two or more)
different 1D subspaces $E$ and  $\mathcal{B}$ is a finite set of these subspaces. The set
$\mathcal{B}$ includes two or more different subspaces $E$. Select a normalized
directional vector $\gamma_E$ for each $E\in \mathcal{B}$. Let $E_{\mathcal{B}}={\rm
Span}\{\gamma_E\, | \, E \in \mathcal{B}\}$. $N^{**}$ is a critical point of $H$ on
$(N^{**}+E_{\mathcal{B}})\cap \mathbb{R}^n$ because $\gamma_E \in \ker (DH)_{**}$ for all
$E\in \mathcal{B}$ and, therefore, $E_{\mathcal{B}}\subset \ker (DH)_{**}$.

Consider a function $H^{1,\max}(N)=\max_{E\in \mathcal{B}}H^*_E(N) $. This function is
strictly convex on $(N^{**}+E_{\mathcal{B}})\cap \mathbb{R}^n$ and $N^{**}$ is its
minimizer on this set. Indeed, in a vicinity of $N^{**}$ in $(N^{**}+E_{\mathcal{B}})\cap
\mathbb{R}^n$ for every $E\in \mathcal{B}$ Equation (\ref{QuadraticH*}) holds. Consider a
direct sum of $k=|\mathcal{B}|$ copies of $E_{\mathcal{B}}$, $E_1\oplus E_2 \oplus \ldots
\oplus E_k$, where all $E_i$ are the copies of $E_{\mathcal{B}}$ equipped by the entropic
inner product $\langle \bullet,\bullet \rangle_{**}$ and the corresponding Euclidean
norm, and the norm of the sum is the maximum of the norm in the summands: $\|x_1\oplus
\ldots \oplus x_k \|= \max \{\|x_1\|,\ldots , \|x_k\|\}$. The following linear map $\psi$
is a surjection $\psi:E_{\mathcal{B}}\to E_1\oplus E_2 \oplus \ldots \oplus E_k$ because
$E_{\mathcal{B}}={\rm Span}\{\gamma_E\, | \, E \in \mathcal{B}\}$ and if all the summands
are zero for $\Delta \in E_{\mathcal{B}}$ then $\Delta=0$:

$$\psi: \Delta \mapsto \bigoplus_{E\in \mathcal{B}} (\Delta- \gamma_E \langle \gamma_E,
\Delta\rangle_{**})$$

Let us mention that on  $E_{\mathcal{B}}$
$$H^{1,\max}(N)-H(N^{**})=\frac{1}{2}\|\psi(\Delta)\|^2 + o(\|\Delta\|^2)$$
 and, therefore, $N^{**}$ is a unique minimizer of $H^{1,\max}(N)$ on $(N^{**}+E_{\mathcal{B}})\cap \mathbb{R}^n$

Because of the local equivalence of the systems with detailed and complex balance and
Theorem~\ref{theor:ComplBalGMALH} we also get the following proposition.
\begin{proposition}\label{Prop:HmaxComBal}
$H^{1, \max}_{\Upsilon}(N)$ is a Lyapunov function in $\mathbb{R}_+^n$ for all kinetic
equations (\ref{KinUrChemRev}) with the given thermodynamic Lyapunov function $H$, the
complex balance condition (\ref{complexbalanceGENKIN}) and the reaction mechanism
(\ref{stoichiometricequation}) with the set of stoichiometric vectors $\Gamma \subseteq
\Upsilon \cup -\Upsilon$.
\end{proposition}
Here, we use the set $\Upsilon \cup -\Upsilon$ instead of just $\Upsilon$ in
Proposition~\ref{Prop:HmaxDetBal} because the direct and reverse reactions are included
in the stoichiometric equations (\ref{stoichiometricequation}) separately.

\subsection*{A2. Forward--invariant peeling}

We use the quasiequilibrium functions and their various combinations for construction of
new Lyapunov functions from the known thermodynamic Lyapunov functions, $H$, and an
arbitrary convex function $F$. In this procedure, we delete some parts from the sublevel
sets of $F$ to make the rest positively--invariant with respect to GMAL kinetics with the
given reaction mechanism and detailed or complex balance. We call these procedures the
{\em forward--invariant peeling}.
\begin{figure}[t]
\centering{
\boxed{\includegraphics[height=0.24\textwidth]{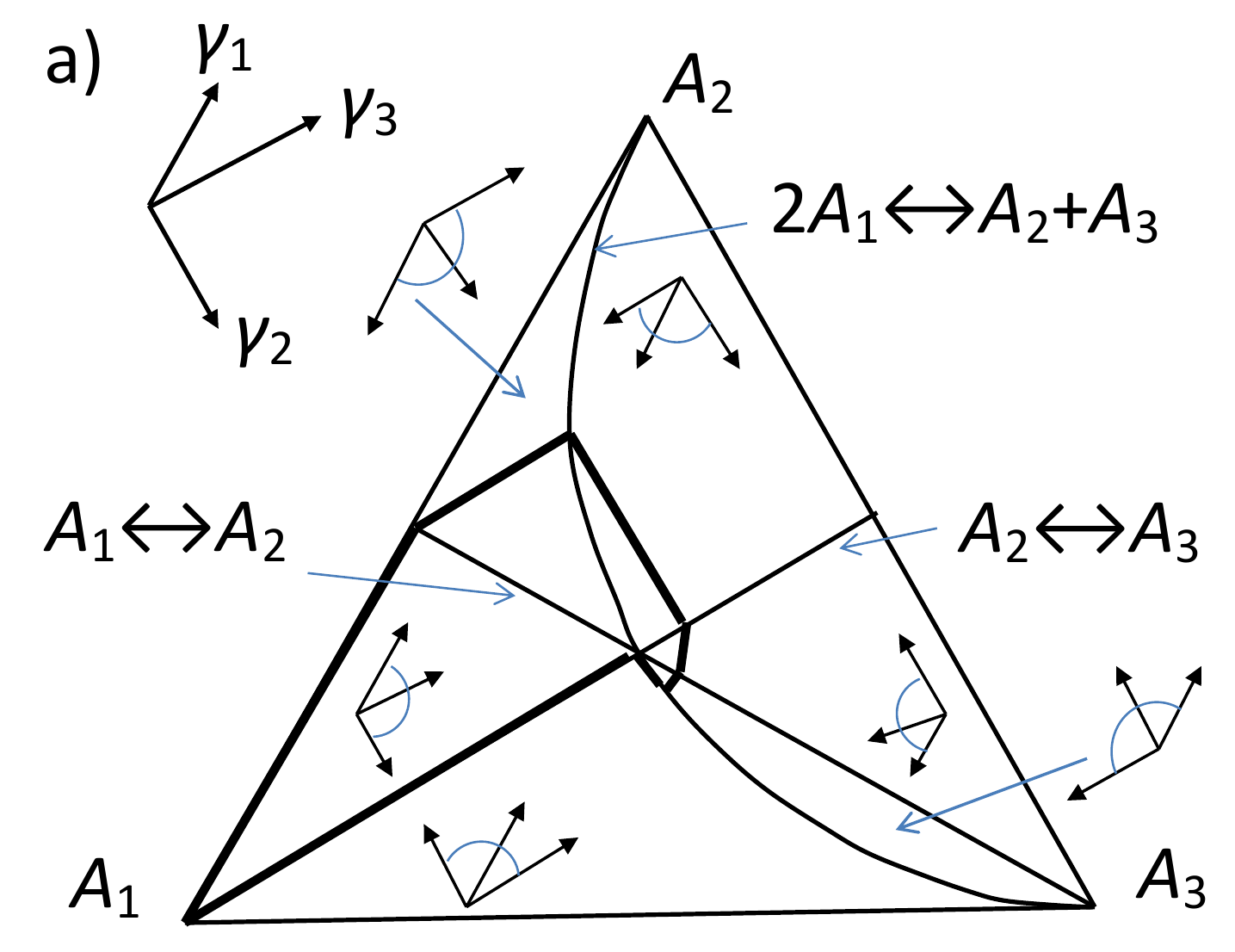}}
\boxed{\includegraphics[height=0.24\textwidth]{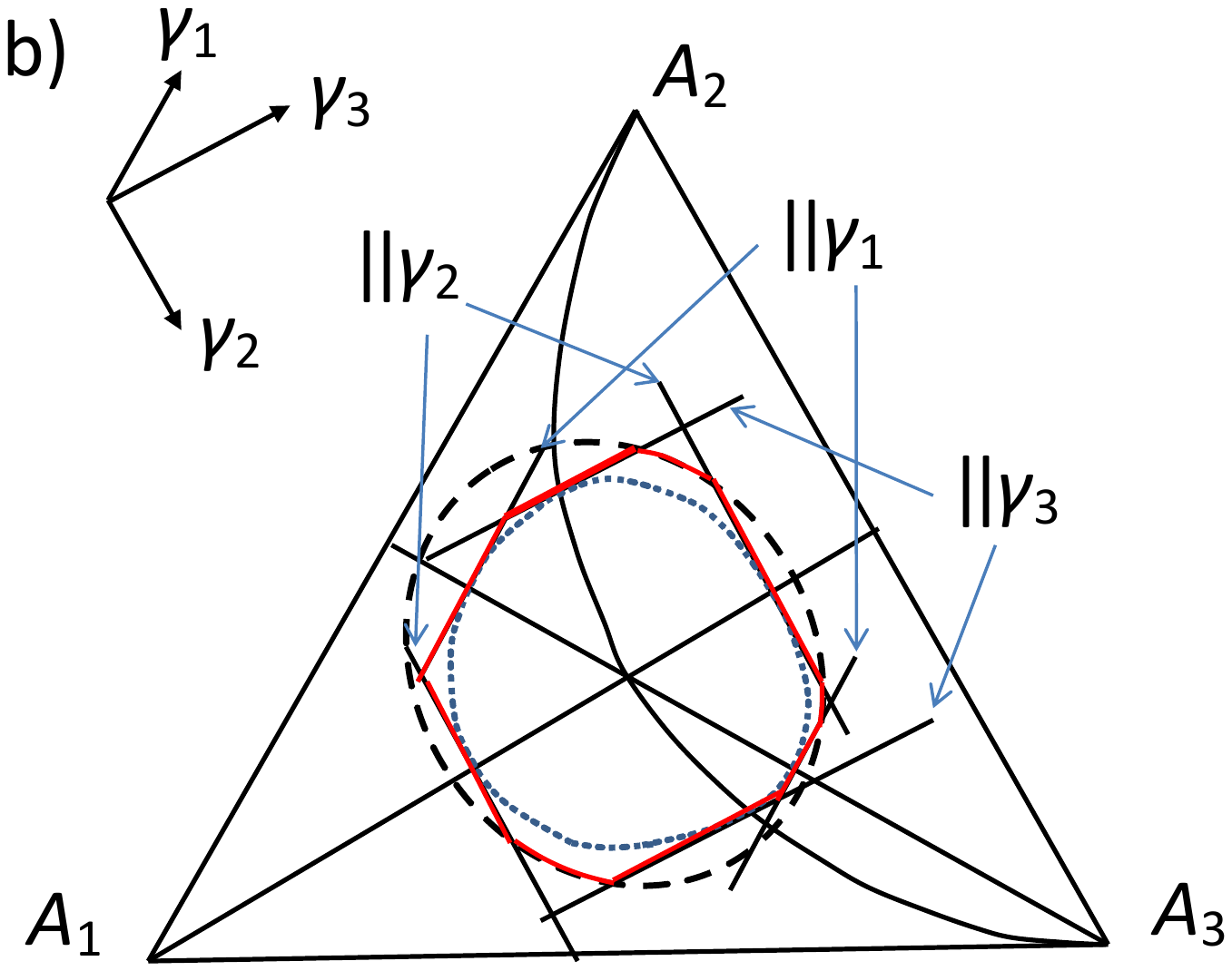}}
\boxed{\includegraphics[height=0.24\textwidth]{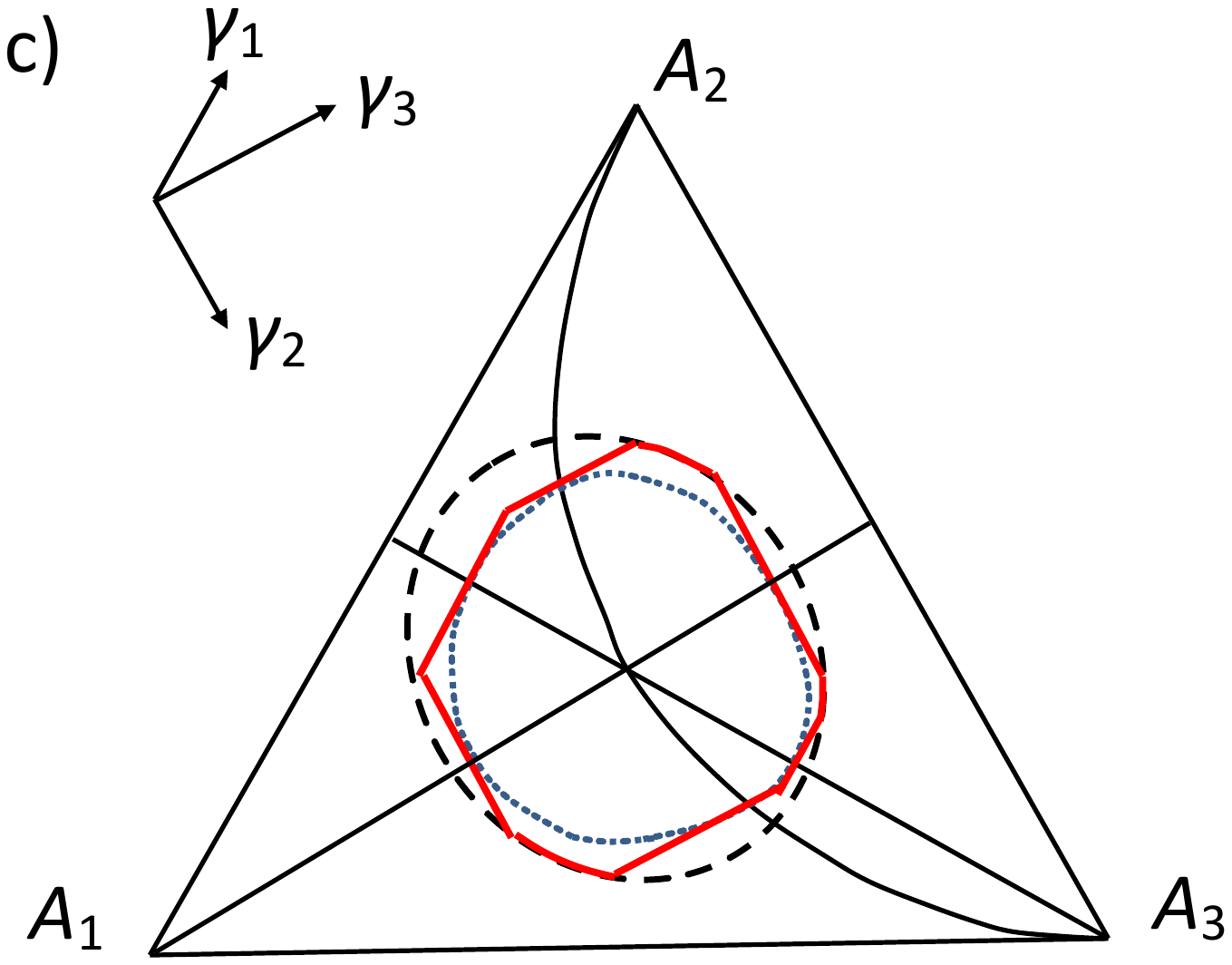}}
\caption{Peeling of convex sets (2D): (a) The reaction mechanism $A_1 \rightleftharpoons A_2$,
$A_2 \rightleftharpoons A_3$, $2A_1 \rightleftharpoons A_2+A_3$. The concentration triangle $c_1+c_2+c_3=b$
is split by the partial equilibria lines into six compartments. In each compartment, the cone of possible directions (the angle) $\mathbf{Q}_{\rm DB}$
is presented. The positively invariant set which includes the $A_1$ vertex is outlined by bold. (b) The set $U$
is outlined by the dashed line, the level set $H=h-\varepsilon$ is shown by the dotted line. The levels of $H^*_{\gamma}$ are the
straight lines $\|\gamma$. The boundary of the peeled set $U^{\varepsilon}_{\{\gamma_1,\gamma_2, \gamma_3\}}$ is shown by red.
(c) The sets $U$, $U^{\varepsilon}_{\{\gamma_1,\gamma_2, \gamma_3\}}$  and  the level set $H=h-\varepsilon$ without auxiliary lines are presented. \label{2Dpeeling}}}
\end{figure}

\begin{figure}[p]
\centering{
\boxed{a) \includegraphics[height=0.37\textwidth]{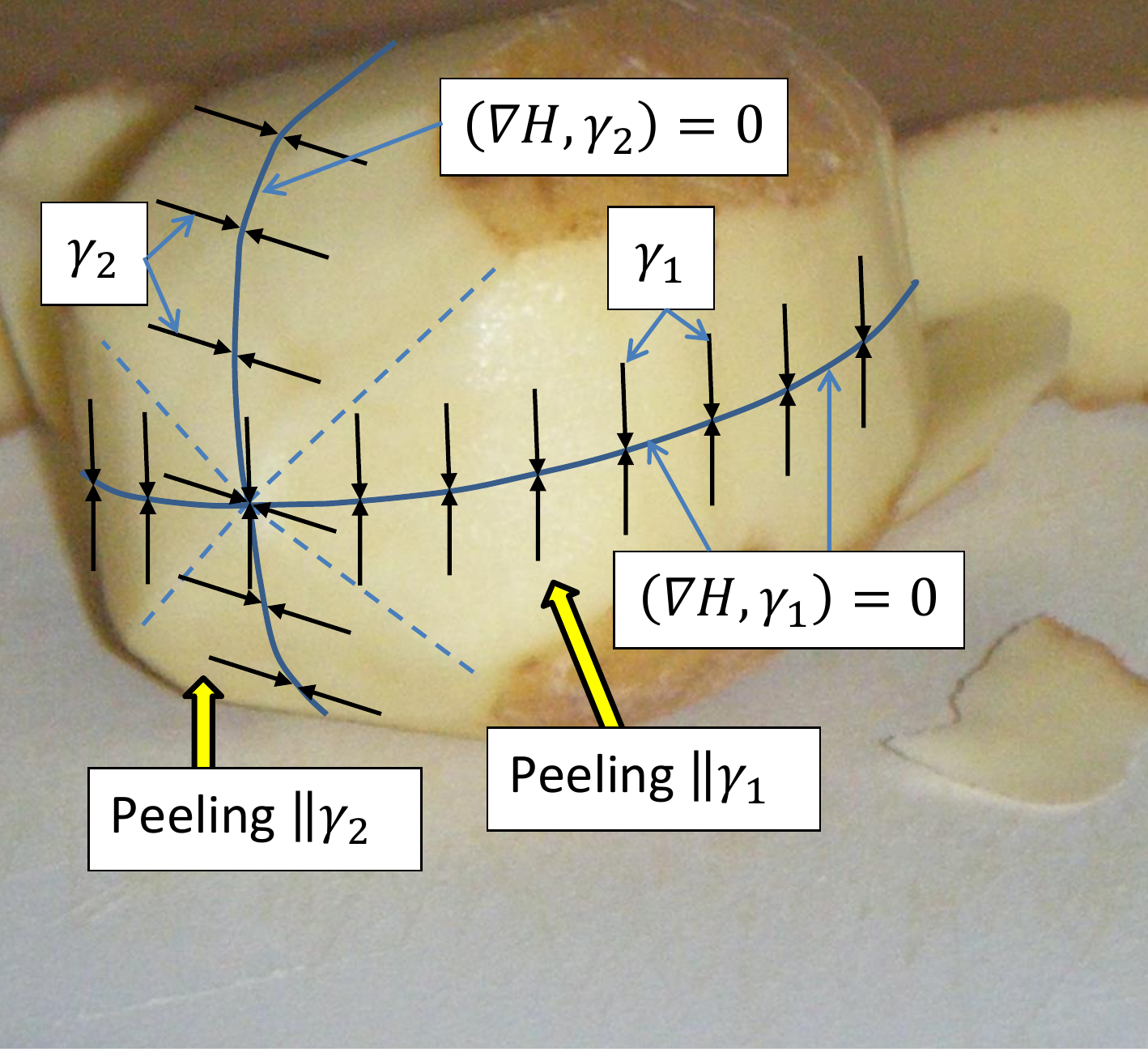}}\hspace{3mm}
\boxed{b) \includegraphics[height=0.37\textwidth]{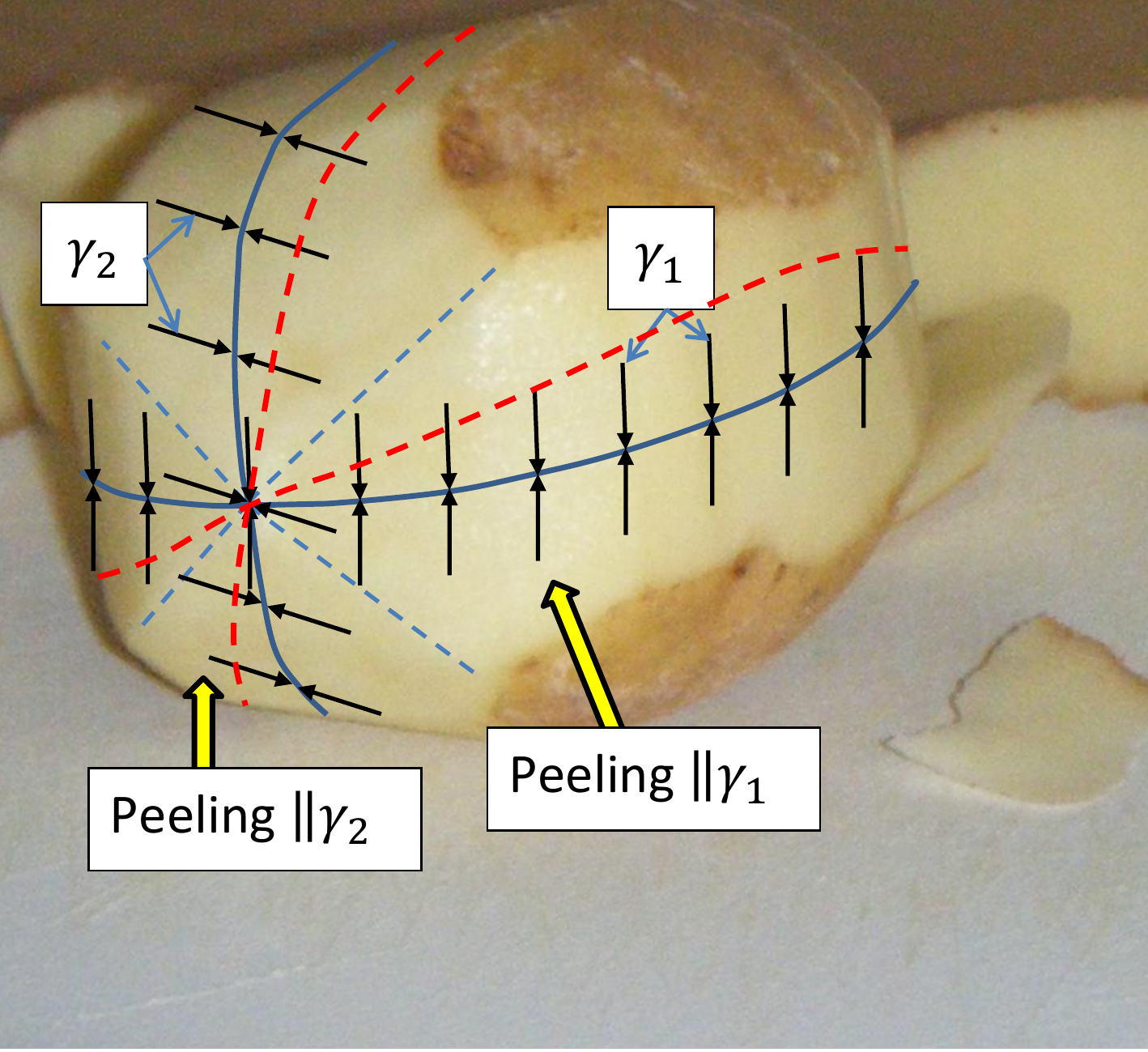}} \\
\boxed{c) \includegraphics[height=0.37\textwidth]{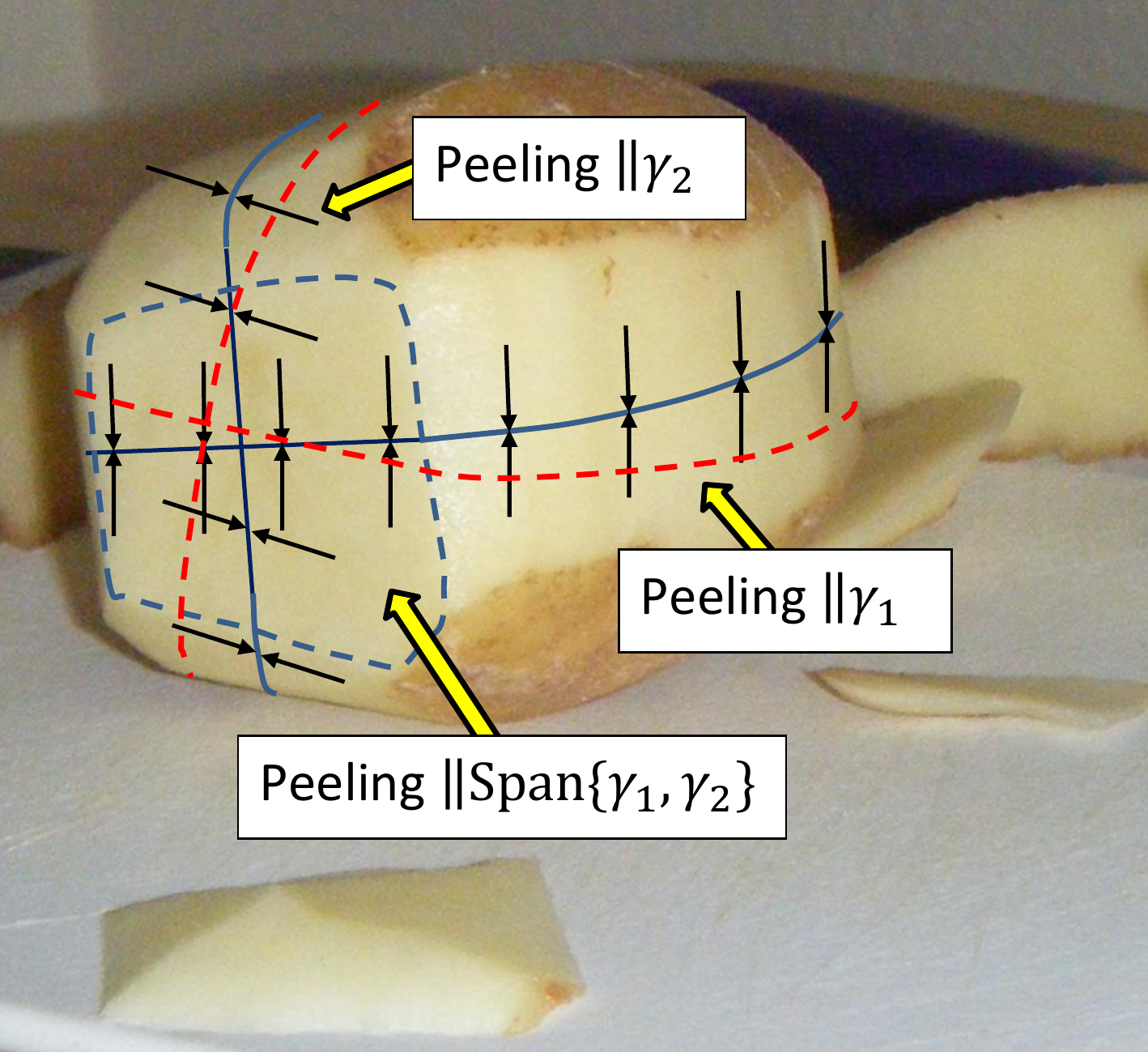}}
\caption{Peeling of convex sets (3D): The unpeeled potato corresponds to the convex set $U$. The partial equilibria $(\nabla H, \gamma_i)=0$ ($i=1,2$)
are presented with the corresponding stoichiometric vectors $\gamma_i$. (a) Near the
intersection of the partial equilibrium surfaces the peeling $\| \gamma_1$ is separated from the peeling $\| \gamma_2$  by the cross of the dashed lines.
(b) After deformation of the partial equilibria (red dashed lines) the peeled set remains forward--invariants if the deformed partial
equilibria for individual reactions do not leave the corresponding peeled 1D faces and, in particular, the intersection of the partial equilibria does not change.
(c) The additional peeling $\| {\rm Span}\{\gamma_1,\gamma_2\}$ makes the peeled set forward--invariant with respect to
the set of systems with interval reaction rate constants. The partial equilibria for any combination of reactions
can move in the limits of the corresponding faces (red dashed lines and their intersection in the Figure, panel c).  \label{Potatopeeling}}}
\end{figure}
Let $U \subset \overline{\mathbb{R}_+^*}$ be a convex compact set of non-negative
$n$-dimensional vectors $N$ and for some $\eta>\min H$ the $\eta$-sublevel set of $H$
belongs to $U$: $\{N \, | \, H(N)\leq \eta\} \subset U $. Let $h> \min H$ be the maximal
value of such $\eta$. Select a {\em thickness of peel}  $\varepsilon>0$.

We define the {\em peeled} set $U$ as
$$U^{\varepsilon}_{\Upsilon}= U \cap \{N  \in \overline{R^n_+} \, | \, H^{1,  max}_{\Upsilon}(N)\leq h-\varepsilon
\}$$ For sufficiently small $\varepsilon>0$ ($\varepsilon< h- \min H$) this set is non-empty and forward--invariant.
\begin{proposition}\label{Prop:Forward-peeledDB}
For sufficiently small $\varepsilon >0$ ($\varepsilon< h- \min H$) the peeled set
$U^{\varepsilon}_{\Upsilon}$ is non-empty. If it is non-empty then it is
forward--invariant with respect to kinetic equations (\ref{KinUrChemRev}) with the
thermodynamic Lyapunov function $H$, reaction rates presented by Equation
(\ref{DBreactionrate}) (detailed balance) and the reversible reaction mechanism
(\ref{reversibleMechanism}) with the set of stoichiometric vectors $\Gamma \subseteq
\Upsilon$.
\end{proposition}
\begin{proposition}\label{Prop:Forward-peeledCB}
For sufficiently small $\varepsilon >0$ ($\varepsilon< h- \min H$) the peeled set
$U^{\varepsilon}_{\Upsilon}$ is non-empty. If it is non-empty then it is
forward--invariant with respect to all kinetic equations (\ref{KinUrChemRev}) with the
given thermodynamic Lyapunov function $H$, the complex balance condition
(\ref{complexbalanceGENKIN}) and the reaction mechanism (\ref{stoichiometricequation})
with any set of stoichiometric vectors $\Gamma \subseteq \Upsilon \cup -\Upsilon$.
\end{proposition}

The forward--invariant peeling for a 2D nonlinear kinetic scheme is demonstrated in
Figure~\ref{2Dpeeling}. A 3D example is presented in Figure~\ref{Potatopeeling}. It is
worth to mention that the peeled froward--invariant sets have 1D faces near the partial
equilibria. These faces are parallel to the stiochiometric vectors of the equilibrating
reactions.

Let $F(N)$ be a continuous strictly convex function with bounded level sets on the
non-negative orthant. For each level of $H$, $h \in {\rm im} H$, we define the  level
$$f(h)=\max_{H(N)=h}F(N)$$. Let $f^*(h) \geq f(h)$ be any strictly increasing function. In particular, we
can take $f^*(h)=f(h)+\varepsilon$ ($\varepsilon>0$).
 Introduce the {\em peeled} function
$$F^{f^*(h)}_{\Upsilon}(N)=\max \{F(N),f^*(H^{1, \max}_{\Upsilon}(N))\}$$
Applying Proposition~\ref{Prop:Forward-peeledDB} to sublevel sets of
$F^{f^*(h)}_{\Upsilon}$ we obtain the following propositions.
\begin{proposition}\label{Prop:Lyap-peeledDB}
$F^{f^*(h)}_{\Upsilon}(N)$  is a Lyapunov function in $\mathbb{R}_+^n$  for all kinetic
equations (\ref{KinUrChemRev}) with the thermodynamic Lyapunov function $H$, reaction
rates presented by Equation (\ref{DBreactionrate}) (detailed balance) and the reversible
reaction mechanism (\ref{reversibleMechanism}) with any set of stoichiometric vectors
$\Gamma \subseteq \Upsilon$.
\end{proposition}
\begin{proposition}
$F^{f^*(h)}_{\Upsilon}(N)$ is a Lyapunov function in $\mathbb{R}_+^n$ for all kinetic
equations (\ref{KinUrChemRev}) with the given thermodynamic Lyapunov function $H$, the
complex balance condition (\ref{complexbalanceGENKIN}) and the reaction mechanism
(\ref{stoichiometricequation}) with the set of stoichiometric vectors $\Gamma \subseteq
\Upsilon \cup -\Upsilon$.
\end{proposition}

The level sets of $F^{f^*(h)}_{\Upsilon}(N)$ have 1D faces parallel to the stoichiometric
vectors $\gamma \in \Upsilon$ near the corresponding partial equilibria outside a
vicinity of the intersections of these partial equilibria hypersurfaces. At the
intersections of two partial equilibria there is a singularity and the size of both faces
tends to zero (Figure~\ref{Potatopeeling}~a).

Let us consider kinetic systems with perturbed thermodynamic potentials $H'=H+\Delta H$,
where $\Delta H$ is uniformly small with it second derivatives. For such the perturbed
systems in a bounded set all the partial equilibria  are close to the partial equilibria
of the original system. (An important case is the perturbation of $H$ by a linear
functional $\Delta H$.) We will modify the peeling procedure to create forward--invariant
sets for sufficiently small perturbations.

Let us look on the forward--invariant peeled set on Figure~\ref{Potatopeeling}~a. If we
slightly deform the partial equilibria surfaces for each reaction
(Figure~\ref{Potatopeeling}~b, red dashed lines) but keep their intersection unchanged
then the peeled set may remain forward--invariant. It is sufficient that the
intersections of the partial equilibria with the border of $U$ in $\mathbb{R}^n$ do not
leave the corresponding 1D faces (Figure~\ref{Potatopeeling}~b). If we perform additional
peeling near the intersections of the partial equilibria (Figure~\ref{Potatopeeling}~c),
then the peeled set may be positively invariant with respect to the kinetic equations
with perturbed thermodynamic potentials (or, even a bit stronger, with respect to kinetic
equations with interval coefficients for sufficiently small intervals).

We consider the reversible reaction mechanism (\ref{reversibleMechanism}) with the set of
the stoichiometric vectors $\Upsilon$. $U \subset \overline{R^n_+}$  is a convex compact
set, and for some $\eta>\min H$ the $\eta$-sublevel set of $H$ belongs to $U$: $\{N \in
\overline{R^n_+}\, | \, H(N)\leq \eta\} \subset U $. Let $h> \min H$ be the maximal value
of such $\eta$.

Select a sequence of thicknesses of peels $\varepsilon_1,\varepsilon_2, \ldots,
\varepsilon_k>0$, where $k= {\rm rank}\Upsilon$. Let us use for the peeling the functions
$H^{i, \max}_{\Upsilon}(N)$ (\ref{Hmax}) ($i=1,\ldots {\rm rank} \Upsilon$).

For each $i$ we consider the sublevel set $$U_i=\{N\in \overline{\mathbb{R}_+^n} \, | \,
H^{i, \max}_{\Upsilon}(N)\leq h- \sum_{j=1}^i \varepsilon_j \}$$ for sufficiently small
numbers $\varepsilon_j>0$ all these sets are non-empty. The peeled $U$ for this sequence
of thicknesses is defined as
\begin{equation}\label{SecondPeeled}
U^{\varepsilon_1, \ldots, \varepsilon _k}_{\Upsilon}= \bigcap_{i=0}^k U_i
\end{equation}
 where we take $U_0=U$.

Definition of $U_k$ ($k= {\rm rank}\Upsilon$) requires some comments. If ${\rm
rank}\Upsilon =n$ then ${\rm Span}(\Upsilon)=\mathbb{R}^n$ and the corresponding
quasiequilibrium $N^*_{\mathbb{R}^n}$ is the global equilibrium, i.e.

$$H^*_{\mathbb{R}^n}=\min_{N\in \mathbb{R}^n_+} H(N)\,(=\min H), \;\;
N^*_{\mathbb{R}^n}=\underset{{N\in \mathbb{R}^n_+}}{\operatorname{argmin}}H(N)$$

In this case, $H^{k, \max}_{\Upsilon}(N)\equiv \min H$ and either $U_k$ is the
nonnegative orthant (if $h- \sum_{j=1}^i \varepsilon_j \geq \min H$) or it is empty (if
$h- \sum_{j=1}^i \varepsilon_j < \min H$). Therefore, in this case the term $U_k$ is not
needed in Equation (\ref{SecondPeeled}).

If $k= {\rm rank}\Upsilon<n$ then the term $U_k$ is necessary. In this case, $H^{k,
\max}_{\Upsilon}(N)=H^*_{{\rm Span}(\Upsilon)}$ and $U_k$ defines non-trivial peeling.
\begin{proposition}\label{Prop:SecPeeling}
\begin{enumerate}
\item For sufficiently small thicknesses $\varepsilon_1, \ldots, \varepsilon _k>0$
    the peeled set $U^{\varepsilon_1, \ldots, \varepsilon _k}_{\Upsilon}$ is
    non-empty and forward--invariant with respect to kinetic equations
    (\ref{KinUrChemRev}) with the thermodynamic Lyapunov function $H$, reaction rates
    presented by Equation (\ref{DBreactionrate}) (detailed balance) and the
    reversible reaction mechanism with the set of stoichiometric vectors $\Gamma
    \subseteq \Upsilon$.
\item For these thicknesses $\varepsilon_1, \ldots, \varepsilon _k>0$ the peeled set
    $U^{\varepsilon_1, \ldots, \varepsilon _k}_{\Upsilon}$ is forward--invariant with
    respect to kinetic equations (\ref{KinUrChemRev}) with the perturbed
    thermodynamic Lyapunov function $H+\Delta H$, reaction rates presented by
    Equation (\ref{DBreactionrate}) (detailed balance) and the reversible reaction
    mechanism with the set of stoichiometric vectors $\Gamma \subseteq \Upsilon$ if
    the perturbation $\Delta H$ is sufficiently uniformly small with its second
    derivatives.
\end{enumerate}
\end{proposition}

The similar proposition is valid for the systems with complex balance (because of the
local equivalence theorem). Peeling of the sublevel sets of a convex function will
produce a Lyapunov function similarly to Proposition~\ref{Prop:Lyap-peeledDB}.

Essential difference of Proposition~\ref{Prop:SecPeeling} from
Proposition~\ref{Prop:Forward-peeledDB} is in the ultimate positive--invariance of
$U^{\varepsilon}_{\Upsilon}$  if it is non-empty. To provide the forward--invariance of
$U^{\varepsilon_1, \ldots, \varepsilon _k}_{\Upsilon}$ we need an additional property.

Let $E\in \mathcal{E}_{\Upsilon}$ be a subspace of the form $E={\rm Span}(\Gamma)$,
$\Gamma \subset \Upsilon$. The {\em quasiequilibrium surface} $\Phi_E \subset
\mathbb{R}^n_+$ is a set of all quasiequilibria $N^*_E(N)$ ($N \in \mathbb{R}^n_+$). The
Legendre transform of $\Phi_E$ (its image in the space of potentials $\check{\mu}$) is
the orthogonal supplement to $E$, for every $\check{\mu}$ from this image
$(\check{\mu},\gamma)=0$, and this is an equivalent definition of $\Phi_E$.

For every $E\in \mathcal{E}_{\Upsilon}$ we define the $E$-faces of $U^{\varepsilon_1,
\ldots, \varepsilon _k}_{\Upsilon}$ as follows. Let $\dim E=i$. Consider the generalized
cylindrical surface with the given value $H^*_E (N)=q$

$$S_E^q=\{N\in R^n_+\, | \, H^*_E (N)=q\}$$

The $E$-faces of $U^{\varepsilon_1, \ldots, \varepsilon _k}_{\Upsilon}$ belong to the
intersection

$$\Psi ^{\varepsilon_1, \ldots, \varepsilon
_k}_{\Upsilon,\, E}= S_E^{h-(\varepsilon_1+\ldots+\varepsilon_i)}\cap U^{\varepsilon_1,
\ldots, \varepsilon _k}_{\Upsilon}$$
\begin{proposition}\label{FaceCondition}Let $h-\sum_{j=1}^k \varepsilon_j > \min H$ and
 for every $E,L\in \mathcal{E}_{\Upsilon}$ the following property holds:
$$\Psi ^{\varepsilon_1, \ldots, \varepsilon _k}_{\Upsilon,\, E} \cap \Phi_L =\emptyset$$
if $L \nsubseteq E$. Then the peeled set $U^{\varepsilon_1, \ldots, \varepsilon
_k}_{\Upsilon}$ is non-empty and forward--invariant with  respect to kinetic equations
(\ref{KinUrChemRev}) with the perturbed thermodynamic Lyapunov function $H+\Delta H$,
reaction rates presented by Equation (\ref{DBreactionrate}) (detailed balance) and the
reversible reaction  mechanism with the set of stoichiometric vectors $\Gamma \subseteq
\Upsilon$ if the perturbation $\Delta H$ is sufficiently uniformly small with its second
derivatives.
\end{proposition}

The proof is an application of the general $H$-theorem based on the partial equilibrium
criterion (for illustration see Figure~\ref{Potatopeeling}~c). The condition of
Proposition~\ref{FaceCondition} means that the result of peeling in higher dimensions
does not destroy the main property of the 1D peeling
(Proposition~\ref{Prop:Forward-peeledDB}): if a positive boundary point $N$ of the bodily
peeled set is a partial equilibrium in direction  $\gamma\in \Upsilon$ then the
stoichiometric vector $\gamma$ belongs to a supporting hyperplane of this peeled set at
$N$.

\subsection*{A3. Greedy peeling}

The goal of the forward--invariant peeling is to create a forward--invariant convex set
from an initial convex set $U$ by deletion (peeling) of its non-necessary parts. The
resulting (peeled) set should be forward--invariant with respect to any GMAL kinetics
with a given reaction mechanism and the thermodynamic Lyapunov function $H$ (``free
energy''). In more general but practically useful settings, we consider not  a single
system but a family of systems with the given reaction mechanism but for a set of
Lyapunov functions $H(N)=H_0(N)+(l,N)$ where $l\in Q$ and $Q\subset \mathbb{R}^n$ is a
convex compact set. We would like to produce a set that is forward--invariant with
respect to all these systems (and, therefore, with respect to the differential inclusion
(compare to (\ref{KinUrChem}))
\begin{equation}\label{KinInclChem}
\frac{\D N}{\D t} \in  V
\sum_{\rho} \gamma_{\rho} \varphi_{\rho}\left[\exp(\alpha_{\rho}, \check{\mu})-\exp(\beta_{\rho}, \check{\mu})\right]
\end{equation}
where
$$\check{\mu}- \nabla H_0(N)=l\in Q$$
and $\varphi_{\rho} \in \mathbb{R}_+ $.

Further on, we consider systems with fixed volume, therefore we omit the factor $V$ and
make no difference between the amounts $N_i$ and the concentrations $c_i$.

The peeling procedure proposed in the previous subsection works but it is often too
extensive and produces not the maximal possible forward--invariant set. We would like to
produce the maximal forward--invariant subset of $U$ and, therefore, have to minimize
peeling. Here we meet a slightly unexpected obstacle. The union (and the closure) of
forward--invariant sets is also forward--invariant, whereas the union of convex sets may
be non-convex. Therefore, there exists the unique maximal forward--invariant subset of
$U$ but it may be non-convex and the maximal convex forward--invariant subset may be
non-unique. If we would like to find the maximal forward--invariant subset then we have
to relax the requirement of convexity.

A  set $U$ is {\em directionally convex} with respect to a set of vectors $\Gamma$ if for
every $x\in U$ and $\gamma \in \Gamma$ the intersection $(x+\mathbb{R} \gamma) \cap U$ is
a segment of a straight line:
$$(x+\mathbb{R} \gamma) \cap U=(x+[a,b] \gamma),\;\mbox{or}\;(x+]a,b] \gamma), \;
\mbox{or}\;(x+[a,b[ \gamma), \mbox{or}\;(x+]a,b[ \gamma)$$

The minimal forward--invariant non-convex (but directionally convex) sets  were
introduced in \cite{Gorban1979} and studied for chemical kinetics in \cite{G11984} and
for Markov chains (master equation) in \cite{Zylka1985}.

Let $U\subset \mathbb{R}^n_+$ be a compact subset. The greedy peeling of $U$  is
constructed for an inclusion (\ref{KinInclChem}) as a sequence of peeling operations
$\Pi_{\gamma}$, where $\gamma$ is a stoichiometric vector of an elementary reaction. A
point $x \in U$ belongs to $\Pi_{\gamma}(U)$ if and only if there exists such a segment
$[a,b] \subset \mathbb{R}$ that
\begin{itemize}
\item $0\in [a,b]$;
\item $x+[a,b] \gamma \subset U$;
\item if $y \in (x+\mathbb{R} \gamma) \cap \mathbb{R}_+^n$ and $(\nabla_N
    H_0(N)\left.\right|_{N=y}+l,\gamma)= 0$ for some $l\in Q$ then $y \in x+[a,b]
    \gamma$.
\end{itemize}

We  call the set
$$S_{\gamma}=\{N\in \mathbb{R}_+^n \, | \, (\nabla     H_0\left.\right|_{N}+l,\gamma)= 0 \mbox{ for some }l\in Q\}$$
the {\em equilibrium strip } for the elementary reaction with the stoichiometric vector
$\gamma$.

Another equivalent description of the operation $\Pi_{\gamma}$ may be useful. Find the
orthogonal projection of $U \cap S_{\gamma}$ onto the orthogonal complement to $\gamma$,
the hyperplane $\gamma^{\bot}\subset\mathbb{ R}^n$. Let $\pi_{\gamma}^{\bot}$ be the
orthogonal projector onto this hyperplane. Find all such $z \in \pi_{\gamma}^{\bot} (U
\cap S_{\gamma})$ that

$$((\pi_{\gamma}^{\bot})^{-1}z) \cap S_{\gamma} =U \cap S_{\gamma}$$

This set is the {\em base} of $\Pi_{\gamma}(U)$, i.e. it is
$$B_{\gamma}(U)=\pi_{\gamma}^{\bot}(\Pi_{\gamma}(U))$$

For each $z \in B_{\gamma}(U)$ consider the straight line $(\pi_{\gamma}^{\bot})^{-1} z$.
This line is parallel to $\gamma$ and its orthogonal projection onto ${\gamma}^{\bot}$ is
one point $z$. The intersection $S_{\gamma}\cap ((\pi_{\gamma}^{\bot})^{-1} z)$ is a
segment. Find the maximal connected part of $U \cap ((\pi_{\gamma}^{\bot})^{-1} z)$ that
includes this segment. This is also a segment (a {\em fiber}). Let us call it
$F_{z,\gamma}(U)$. We define $$\Pi_{\gamma}(U)=\bigcup_{z\in B_{\gamma}(U)}
F_{z,\gamma}(U) $$

The set $\Pi_{\gamma}(U)$ is forward--invariant with respect to the differential
inclusion (\ref{KinInclChem}) if the reaction mechanism consists of one reaction with the
stoichiometric vector $\gamma$. It is directionally convex in the direction $\gamma$. Of
course, if we apply the operation $\Pi_{\gamma'}$ with a different stoichiometric vector
$\gamma'$  to $\Pi_{\gamma}(U)$ then the forward-invariance with respect to the
differential inclusion (\ref{KinInclChem}) for one reaction with the previous
stoichiometric vector $\gamma$ may be destroyed. Nevertheless, if we apply an infinite
sequence of operations $\Pi_{\gamma_{\rho}}$ (${\rho}=1, \ldots, m$) where all the
stoichiometric vectors $\gamma_{\rho}$ from the reaction mechanism appear infinitely many
times then the sequence converges to the maximal forward--invariant subset of $U$ because
of monotonicity (in particular, this limit may be empty if there is no positively
invariant subset in $U$). The limit set is directionally convex in directions
${\gamma_{\rho}}$ (${\rho}=1, \ldots, m$) and is the same for all such sequences.

\subsection*{A4. A toy example}

Let us consider a reaction mechanism

\begin{equation}\label{3Dtoy}
A_1{ \overset{k_{1}}{\rightarrow}}  A_2 { \overset{k_{2}}{\rightarrow}} A_3 {
\overset{k_{3}}{\rightarrow}} A_1, \;\;\; 2A_1 \underset{k_{-4}}{ \overset{k_{4}}{\rightleftharpoons}} 3A_2
\end{equation}
with the classical mass action law and interval constants $0< k_{i \, \min} \leq k_i \leq
k_{i \, \max}< \infty$. Consider the kinetic equations with such interval constants and
classical mass action law.

The stoichiometric vectors of the reactions are
\begin{equation}
\gamma_1=\left(
\begin{array}{c}
-1\\ 1\\ 0
\end{array}
\right);\;
\gamma_2=\left(
\begin{array}{c}
0\\ -1\\ 1
\end{array}
\right);\;
\gamma_3=\left(
\begin{array}{c}
1\\ 0\\ -1
\end{array}
\right);\;
\gamma_4=\left(
\begin{array}{c}
-2\\ 3\\ 0
\end{array}
\right);\;
\end{equation}

We will demonstrate how to use peeling for solving of the following problem for the
system (\ref{3Dtoy}): is it possible that the solution of the differential inclusion with
these interval constants starting from a positive vector will go to zero when $t\to
\infty$? (This question for this system was considered recently as an unsolved problem
\cite{GopaShiu2013}.)

Let us use the local equivalence of systems with complex and detailed balance and
represent this system as a particular case of differential inclusion (\ref{KinInclChem})
(with possible extension of the interval of constants).

The equilibrium concentrations $c_i^*$ in the irreversible cycle satisfy the following
identities:
$$k_1c_1^*=k_2c_2^*=k_3c_3^*,\;\; \frac{c_i^*}{c^*_j}=\frac{k_j}{k_i}$$

Instead of the irreversible cycle of linear reactions we will take the reversible cycle
\begin{equation}\label{reversibleCycle}
A_1\underset{\kappa_{-1}}{ \overset{\kappa_{1}}{\rightleftharpoons}}  A_2
\underset{\kappa_{-2}}{ \overset{\kappa_{2}}{\rightleftharpoons}} A_3
\underset{\kappa_{-3}}{ \overset{\kappa_{3}}{\rightleftharpoons}} A_1
\end{equation}
with the interval restrictions on the {\em equilibrium constants} (the ratios of the
reaction rate constants $\kappa_j/\kappa_{-j}$)
\begin{equation}\label{KappaRestrictions}
 \frac{\min k_2}{\max k_1} \leq \frac{\kappa_1}{\kappa_{-1}} \leq \frac{\max
 k_2}{\min{k_1}}, \;\;
 \frac{\min k_3}{\max k_2} \leq \frac{\kappa_2}{\kappa_{-2}} \leq \frac{\max
 k_3}{\min{k_2}}, \;\;
 \frac{\min k_1}{\max k_3} \leq \frac{\kappa_3}{\kappa_{-3}} \leq \frac{\max
 k_1}{\min{k_3}}
 \end{equation}
The detailed balance condition should also hold for the constants $\kappa_{\pm j}$:
\begin{equation}\label{Kappadetbal}
\kappa_1 \kappa_2 \kappa_3= \kappa_{-1} \kappa_{-2} \kappa_{-3}
\end{equation}

The equilibria for this cycle satisfy the conditions
$$\kappa_{1}c_1^*=\kappa_{-1}c_2^*, \; \kappa_{2}c_2^*=\kappa_{-2}c_3^*, \;
\kappa_{3}c_3^*=\kappa_{-3}c_1^*$$.

\begin{figure}[t]
\centering{\includegraphics[height=0.5\textwidth]{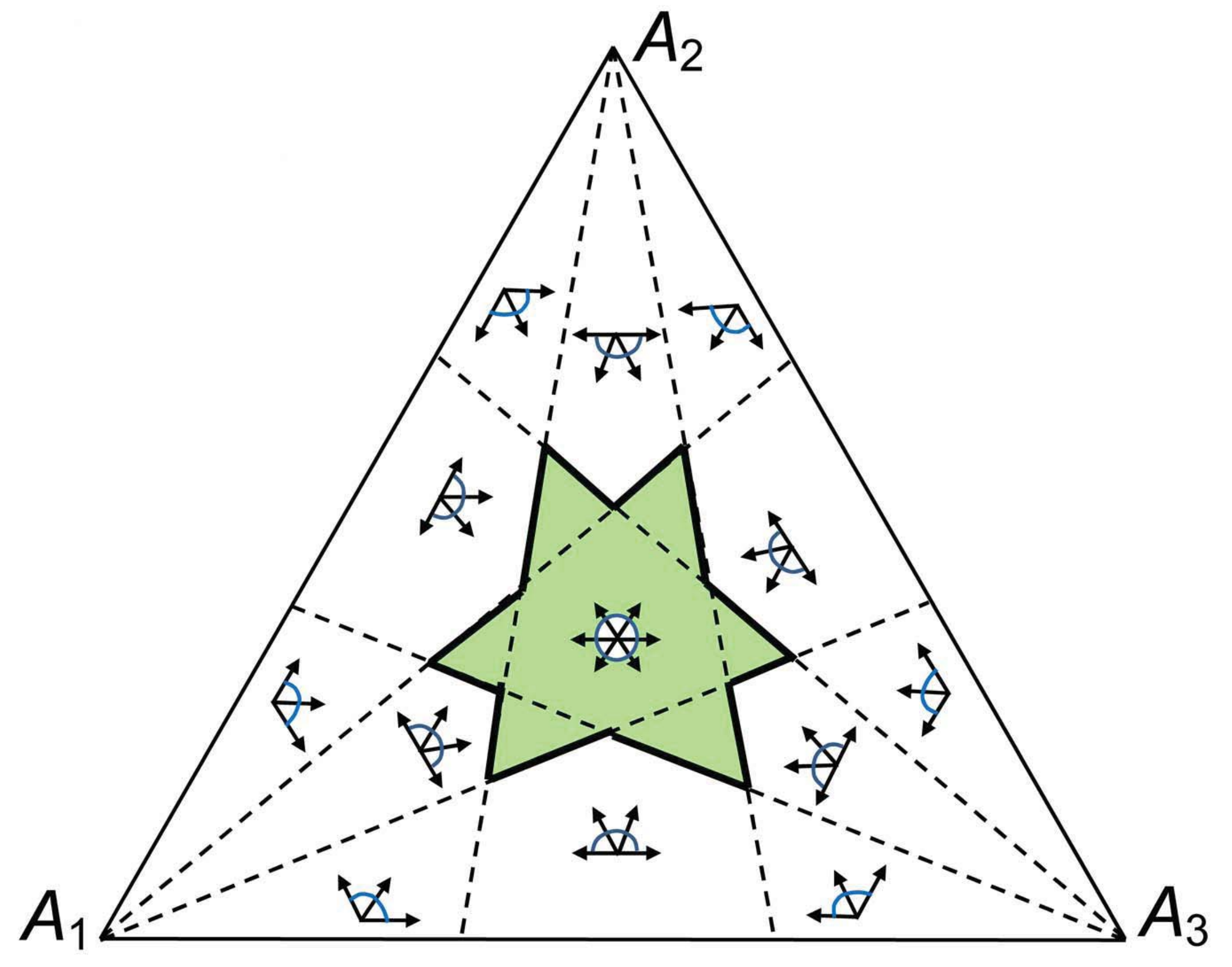}}
\caption{Partial equilibria of the reversible cycle (\ref{reversibleCycle})
with the interval restrictions on the equilibrium constants. The triangle is split by the lines of partial
equilibria $A_i \rightleftharpoons A_j$ into several compartments. The borders of these compartments
are combined from the segments of the dashed lines. These dashed lines correspond to
the minima and maxima of the equilibrium constants $\kappa_{j}/\kappa_{-j}$.  In each compartment, the
cone (the angle) of possible directions of $\dot{c}$ is given. This is a proper cone (an angle that is less than $\pi$)
outside the equilibrium strips, a halfplane in an equilibrium strip of a single reaction, and a
whole plane in the intersection of two such strips. The area of the possible equilibria
(where the angle of possible directions of $\dot{c}$ is the whole plane) is outlined by bold
line and colored in green. \label{IntervalLinearCycle}}
\end{figure}

These conditions provide the same range of  equilibrium concentrations for the reversible
and irreversible cycles. Therefore, the possible value of $\dot{c}$ for the irreversible
cycle in the given interval of reaction rate constants always belongs to the cone of
possible values of $\dot{c}$ of the reversible cycle under the given restrictions
(\ref{KappaRestrictions}) and the detailed balance condition (\ref{Kappadetbal}).

For the reversible cycle the reaction rates are
$$r_1=\kappa_1 c_1-\kappa_{-1}c_2, \; r_2=\kappa_2 c_2- \kappa_{-2}c_3,  \; r_3= \kappa_3 c_3-\kappa_{-3}c_1$$
The reaction rate of the reaction $2A_1 \rightleftharpoons 3A_2$ is $r_4=k_4
c_1^2-k_{-4}c_3^3$.

The time derivatives of the concentrations are
\begin{equation}\label{toyKinur}
\dot{c}_1=-r_1+r_3-2r_4, \; \dot{c}_2=r_1-r_2+3r_4, \; \dot{c}_3=r_2-r_3
\end{equation}

The differential inclusion for the reversible linear cycle (\ref{reversibleCycle}) is
represented in Fig.~\ref{IntervalLinearCycle}. There are three types of areas: (i) area
where the equilibria may be located and the direction of $\dot{c}$ may coincide with any
vector of the linear subspace $\sum_i \dot{c}_i=0$, (ii) areas where direction of one
reaction is indefinite but the signs of two other  reactions rates are fixed, and (iii)
areas where the signs of all reaction rates are fixed. The cones (angles) of possible
vectors $\dot{c}$ are drawn in Fig.~\ref{IntervalLinearCycle}

For the linear system the scheme presented in Figure~\ref{IntervalLinearCycle} does not
depend on the positive value of the balance $\sum_i c_i = \varepsilon$. We can just
rescale $c_i \leftarrow c_i/\varepsilon$ and return to the unit triangle with the unit
sum of $c_i$. The situation is different for the nonlinear reaction $2 A_1
\rightleftharpoons 3 A_2$. Consider the ``equilibrium strip'' where the reaction rate
$r_4=k_4 c_1^2-k_{-4}c_2^3$ may be zero for the admissible reaction rate  constants:
$$\frac{\min k_{-4}}{\max
k_4}\leq\frac{c_1^2}{c_2^3}\leq \frac{\max k_{-4}}{\min k_4}$$
 Let us take this strip on
the plane $\sum_i c_i =\varepsilon$ and return it to the unit triangle by rescaling
($c_i\leftarrow c_i/ \varepsilon$). For small $\varepsilon$ this strip approaches the
$[A_2,A_3]$ edge of the triangle. It is situated between the line
$${c_1}=\sqrt{\varepsilon}\sqrt{\frac{\max k_{-4}}{\min k_4}}(1-{c_3})^{3/2}$$ and the segment
$[A_2,A_3]$. Further we use the notation $\vartheta$ for the coefficient in this formula:
$$\vartheta=\sqrt{\varepsilon}\sqrt{\frac{\max k_{-4}}{\min k_4}}$$
The line
\begin{equation}\label{3/2separator}
c_1=\vartheta (1-c_3)^{3/2}
\end{equation}
separates the equilibrium strip of the reaction $2 A_1 \rightleftharpoons 3 A_2$ (where
$r_4=0$ for some admissible combinations of the reaction rate constants) from the area
where $r_4>0$ (i.e. $k_4 ({\varepsilon c_1})^2-k_{-4}(\varepsilon c_2)^3>0$ for all
admissible $k_4, k_{-4}$. (We use the rescaling from the triangle with $\sum
c_i=\varepsilon$ to the unit triangle without further comments.)

We will study intersection of the equilibrium strip for the reaction $2 A_1
\rightleftharpoons 3 A_2$ with different planes and then scale the result to the balance
plane $\sum_i c_i =1$. The projection of the strip from all planes $\sum_i c_i = a
\varepsilon$ onto the unit triangle for $a\in [\min a,\max a]>0$ belong to the projection
of the strip from the plane $\sum_i c_i = \varepsilon $ with the extended range of the
equilibrium constants:
\begin{equation}\label{rescalingK4}
 \min a \frac{\min k_{-4}}{\max k_4}  \geq \frac{k_{-4}}{k_4}\leq \max a \frac{\max k_{-4}}{\min k_4}
\end{equation}
This rescaling does not cause any difficulty but requires additional check at the end of
construction: does the set of the constructed faces (``peels'') has the bounded ratio
$$\frac{\max {\sum_i c_i}}{\min {\sum_i c_i}}$$
with the upper estimate does not dependent on the values of $\frac{k_{-4}}{k_4}$.

This line is tangent to the segment at the vertex $A_3$ (Fig.~\ref{CurveStrip}). On the
other side of the line the time derivative of $\sum_i c_i$ is positive:

$$\sum_i \dot{c}_i= r_4>0$$

\begin{figure} \caption{The equilibrium strip of the reaction $2 A_1 \rightleftharpoons 3 A_2$ (yellow)
and the area where $\sum_i \dot{c}_i>0$ (blue) rescaled from the triangle with $\sum c_i=\varepsilon$
to the unit triangle  \label{CurveStrip}}
\centering{\includegraphics[height=0.4\textwidth]{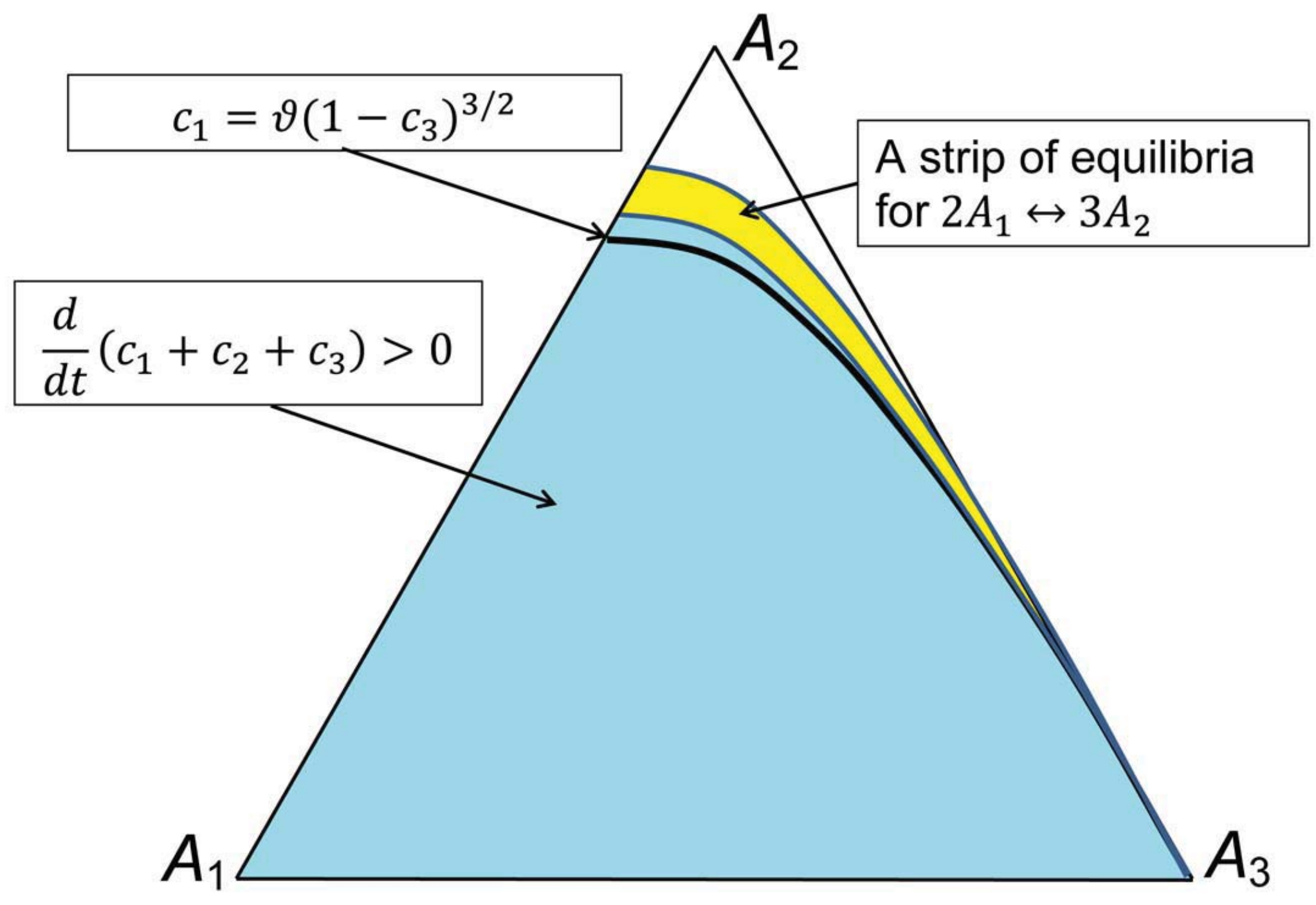}}\end{figure}

\begin{figure} \caption{Faces of the peeled invariant set in the central projection onto unit
triangle. The borders between faces are highlighted by bold.  \label{FaceStructure}}
\centering{\includegraphics[width=0.9\textwidth]{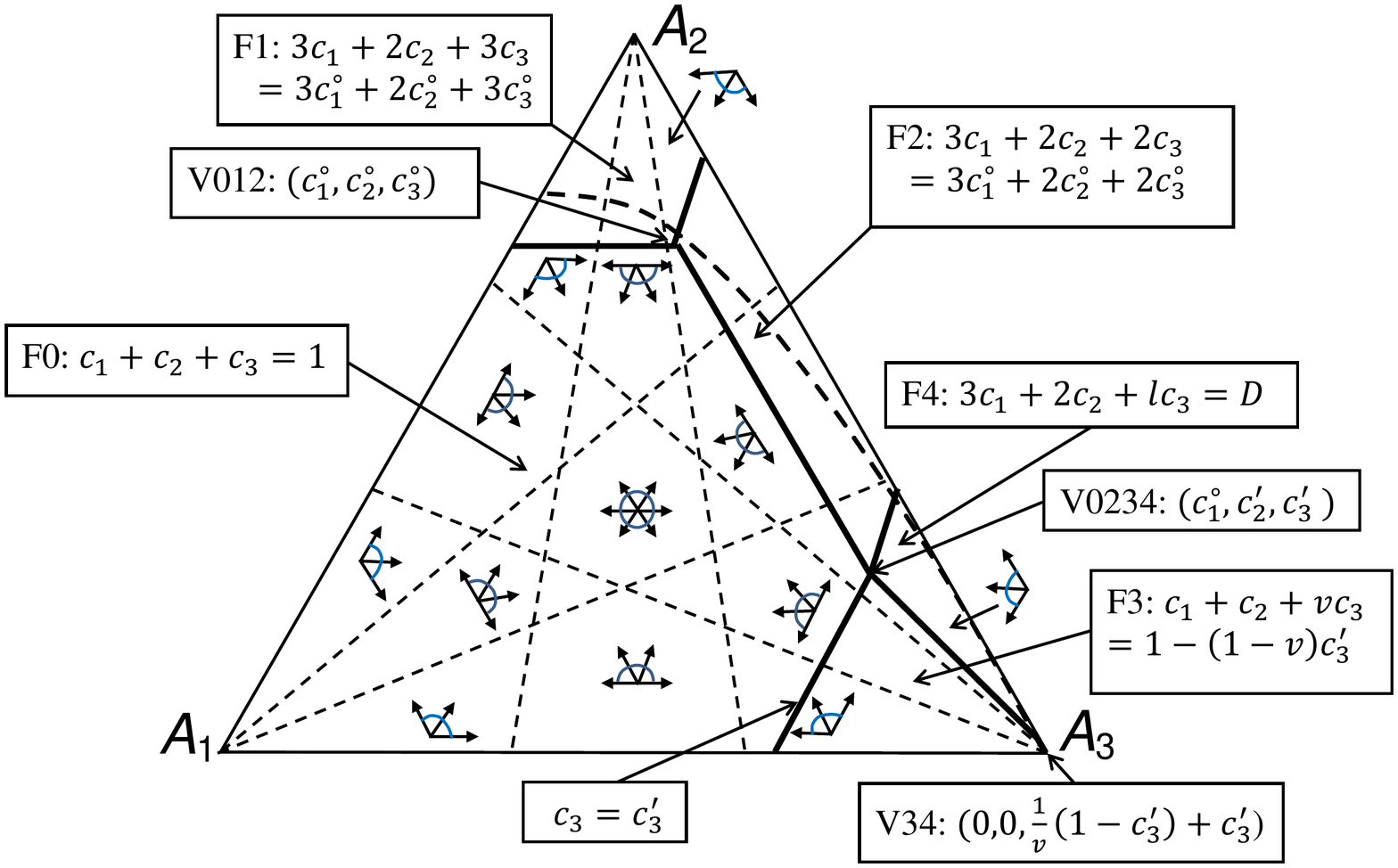}}\end{figure}

Let us describe first the structure of the peeled set. Select for peeling the set $U=\{c
\, | \, \sum_i c_i \geq \varepsilon, \, c_i \geq 0\}$. The structure of peeling scaled to
$c_1+c_2+c_3=1$ is presented in Fig.~\ref{FaceStructure}. It appears that the piecewise
linear peeling is sufficient. There are five faces different from the coordinate planes.
The face F0 is a polygon on the plane $\sum c_i=1$. The face F1 is situated at the $A_2$
corner. It is produced by the peeling parallel to ${\rm Span}\{\gamma_3, \gamma_4\}$. The
plane of F1 is given by the equation $3c_1+2c_2+3c_3=const$. The face F2 is presented by
a parallelogram at the middle of the edge $[A_2,A_3]$ (Fig.~\ref{FaceStructure}). It
covers the intersection of the equilibrium strips of the reactions $2A_1
\rightleftharpoons 3A_2$ and the reaction $A_2 \rightleftharpoons A_3$. F2 is produced by
the peeling parallel to ${\rm Span}\{\gamma_2, \gamma_4\}$. The plane is given by the
equation $3c_1+2c_2+2c_3=const$. Its intersection with the plane $c_1+c_2+c_3=1$ is a
straight line $c_1=c_1^{\circ}, c_2+c_3=1-c_1^{\circ}$ for a sufficiently small
$c_1^{\circ}>0$.

The final fragment of peeling is situated near the vertex $A_3$
(Fig.~\ref{FaceStructure}). It consists of two triangles. The first (F3) is a fragment of
a plane $c_1+c_2+vc_3=const$ ($0<v<1$). Parameter $v$ is defined from the condition of
positive invariance below.

The second triangle (F4) situated near the vertex $A_3$ is parallel to $\gamma_4$ and has
the common edge with F3. The general   plane parallel to $\gamma_4$ is given by the
equation $3c_1+2c_2+lc_3=D$. We will define the parameters $l$ and $D$ using  the
vertices of the face F3, V34 and V0234 (see Fig.~\ref{FaceStructure}).

Let us define the parameters of this peeling. At the $A_2$ corner the peeling is parallel
to ${\rm Span}\{\gamma_3, \gamma_4\}$. The plane can be given by the equation
$3c_1+2c_2+3c_3=const$. The edge between this face and the face $\sum c_i = 1$ belongs to
the straight line $c_2=c_2^{\circ}$, $c_1+c_3=1-c_2^{\circ}$. The level $c_2^{\circ}$
should be selected above all the equilibria of the linear reactions
(Fig.~\ref{IntervalLinearCycle}) but below the intersection of the curve
(\ref{3/2separator}) with the right border of the equilibrium strip of the reaction
$A_1\rightleftharpoons A_3$ given by the equation $c_3 = c_1
\max\left\{\frac{\kappa_{-3}}{\kappa_3}\right\}$. For the intersection we have
$$c_3=\vartheta \max\left\{\frac{\kappa_{-3}}{\kappa_3}\right\}(1-c_3)^{3/2} $$
Therefore, at this point
$$c_3< \vartheta \max\left\{\frac{\kappa_{-3}}{\kappa_3}\right\}$$
and $c_1 <\vartheta $ on the line (\ref{3/2separator}). Therefore, we can select
$$c_2^{\circ}=1-\vartheta \left(\max\left\{\frac{\kappa_{-3}}{\kappa_3}\right\}+1
\right)$$ This $c_2^{\circ}$ is smaller than the value of $c_2$ at the intersection, and
for sufficiently small $\vartheta$ the line $c_2=c_2^{\circ}$ is close to the vertex
$A_2$ and does not intersect the area of possible equilibria of linear reactions (the
area colored in green in Fig.~\ref{IntervalLinearCycle}).

Consider intersection of the straight line $c_2=c_2^{\circ}$, $c_1+c_3=1-c_2^{\circ}$
with the curve (\ref{3/2separator}) and evaluate the value of $c_1$ at this intersection
from above: $c_1=\vartheta (c_2^{\circ}+c_1)^{3/2}$, $c_1<\vartheta$, hence,
$c_1<c_1^{\circ} =\vartheta (c_2^{\circ}+\vartheta)^{3/2}$.

Thus, the vertex V012 at the intersection of three faces, F0, F1, and F2 is selected as
$(c_1^{\circ},c_2^{\circ},c_3^{\circ})$, where
$$c_1^{\circ}=\vartheta \left(1-\vartheta \max\left\{\frac{\kappa_{-3}}{\kappa_3}\right\}\right)^{3/2}$$
$$c_2^{\circ}=1-\vartheta \left(\max\left\{\frac{\kappa_{-3}}{\kappa_3}\right\}+1
\right)$$
$$c_3^{\circ}=1-c_1^{\circ}-c_2^{\circ}=\vartheta\left(1+\max\left\{\frac{\kappa_{-3}}{\kappa_3}\right\}
- \left(1-\vartheta \max\left\{\frac{\kappa_{-3}}{\kappa_3}\right\}\right)^{3/2}\right)$$

To check that this point is outside the equilibrium strip of the reaction
$A_1\rightleftharpoons A_3$, we calculate

$$\frac{c_3^{\circ}}{c_1^{\circ}}=\frac{1+\max\left\{\frac{\kappa_{-3}}{\kappa_3}\right\}}{\left(1-\vartheta
\max\left\{\frac{\kappa_{-3}}{\kappa_3}\right\}\right)^{3/2}} -1>\max\left\{\frac{\kappa_{-3}}{\kappa_3}\right\}$$

The next group of parameters we have to identify are the coordinates of the vertex V0234
$(c_1',c_2',c_3')$ at the intersection of four faces F0, F2, F3, and F4. We will define
it as the intersection of F0, F2, and F3 and then use its coordinates for defining the
parameters of F4. One coordinate, $c_1'$ is, obviously, $c_1'=c_1^{\circ}$ because the
intersection of F2 and F0 is parallel to $\gamma_2$, i.e. it is parallel to the edge
$[A_2,A_3]$ of the unit triangle and $c_1$ is constant on this edge. Another coordinate,
$c_3'$ can be easily determined from the condition that the line $c_3=c_3'$ in the unit
triangle should not intersect the strips of equilibria for the reactions $A_2
\rightleftharpoons A_3$ and $A_1 \rightleftharpoons A_3$. Immediately, these condition
give the inequalities that should hold for all admissible reaction rate constants:
$$c_3'>\frac{\kappa_{-3}}{\kappa_3+\kappa_{-3}}, \;\; c_3'>\frac{\kappa_{2}}{\kappa_2+\kappa_{-2}}$$
Finally,
$$c_3'>\max\left\{\frac{1}{\min\left\{\frac{\kappa_3}{\kappa_{-3}}\right\}+1},\; \frac{1}{1+\min\left\{\frac{\kappa_{-2}}{\kappa_{2}}\right\}}\right\} $$

We can take $c_3'$ between this maximum and 1: for example, we propose
$$c_3'=\frac{1}{2}+\frac{1}{2}\max\left\{\frac{1}{\min\left\{\frac{\kappa_3}{\kappa_{-3}}\right\}+1},\;
\frac{1}{1+\min\left\{\frac{\kappa_{-2}}{\kappa_{2}}\right\}}\right\} $$ For sufficiently
small $\vartheta$, the inequality  $c_3'+c_1^{\circ}<1$ holds, and we can take
$c_2'=1-c_3'-c_1^{\circ}>0$.

If we know $c_3'$ and $v$ then we know the equation of the plane F4:
$$c_1+c_2+vc_3=1-(1-v)c_3'$$

We also find immediately the coordinates of the vertex V34, the intersection of F3 (and
F4) with the coordinate axis $A_3$. This vertex is $(0,0,\frac{1}{v}(1-c_3')+c_3')$.

Let us define the parameters $l$ and $D$ for the face F4. This face should include the
vertices V0234 $(c_1^{\circ},c_2',c_3')$ and V34 $(0,0,\frac{1}{v}(1-c_3')+c_3')$.
Therefore,
$$l=v\left(2+\frac{c_1^{\circ}}{c_1^{\circ}+c_2'}\right), \; D=3c_1^{\circ}+2c_2'+lc_3'$$

To demonstrate the positive invariance of the peeled set we have to evaluate the sign of
the inner product of $\dot{c}$ onto the inner normals to the faces on the faces.

The signs of some reaction rates are unambiguously defined on the faces:
\begin{itemize}
\item On F0 $r_4>0$;
\item On F1 $r_1<0$, and $r_2>0$;
\item On F2 $r_1<0$, and $r_3>0$;
\item On F3 $r_2<0$, $r_3>0$, and $r_4>0$;
\item On F4 $r_1<0$, $r_2<0$, and $r_3>0$.
\end{itemize}

The inner products of $\dot{c}$ (\ref{toyKinur}) onto the inner normals to the faces are:

\begin{itemize}
\item On F0 $\frac{d}{dt}(c_1+c_2+c_3)=r_4>0$;
\item On F1 $\frac{d}{dt}(3c_1+2c_2+3c_3)=-r_1+r_2>0$;
\item On F2 $\frac{d}{dt}(3c_1+2c_2+2c_3)=-r_1+r_3>0$;
\item On F3 $\frac{d}{dt}(c_1+c_2+vc_3)=(1-v)(-r_2+r_3)+r_4>0$ ($0<v<1$);
\item On F4 $\frac{d}{dt}(3c_1+2c_2+lc_3)=-r_1-(2-l)r_2+(3-l)r_3<0$ if $0<l<2$.
\end{itemize}
Thus, the peeled set is positively invariant if $0<l<2$. This means
$$0<v<\frac{1}{1+\frac{c_1^{\circ}}{2(c_1^{\circ}+c_2')}}$$
It is sufficient to take $0<v\leq\frac{2}{3}$ (for example, $v=\frac{2}{3}$) because of
the obvious inequality, $\frac{c_1^{\circ}}{2(c_1^{\circ}+c_2')}<\frac{1}{2}$.

We see that the peeled faces are located between the planes $\sum_i c_i =\varepsilon$ and
$\sum_i c_i =\frac{3}{2}\varepsilon$ (for $v=2/3$). Therefore, it is sufficient to take
in the rescaling (\ref{rescalingK4})  the constants $\max a=\frac{3}{2}$, $\min a =1$
which do not depend on the equilibrium constant.

We have demonstrated that for any  given range of positive kinetic constants any positive
solution of the kinetic inclusion for the system (\ref{3Dtoy}) cannot approach the origin
when $t\to \infty$.

We have started from a system (\ref{3Dtoy}) with interval rate constants and have
embedded the corresponding differential inclusion into a differential inclusion for a
reversible system with detailed balance (\ref{Kappadetbal}) and interval restrictions
onto equilibrium constants (\ref{KappaRestrictions}).

We have constructed a piecewise-linear surface that isolated the $\varepsilon$-vicinity
of the origin from the outside for sufficiently small $\varepsilon>0$. This surface
cannot be intersected by the solutions of the kinetic inclusion in the motion from the
outside to the origin.

The peeling procedure used in this toy-example differs from the universal greedy peeling.
(It is the simplified version of the greedy peeling.) We have guessed the structure of
the corner near $A_3$ and build two plain faces, F3 and F4, instead of a sequence of the
curvilinear ``cylindric'' faces. This piecewise peeling is not minimal but is simpler for
drawing.

\subsection*{Acknowledgement} I am very grateful to Dr Anne Shiu from the  Department
of Mathematics at the University of Chicago. She ensured me that my results annotated in
1979 \cite{GorbanRKCL1980} may be still of interest for the chemical dynamics community.
In this Appendix, I explain one of the methods (forward--invariant peeling) used in this
work, the further details will follow.

\end{document}